\documentclass[11pt]{amsart}
\usepackage[unicode=true,linktoc=pages]{hyperref}

\usepackage{amssymb,amsthm,amscd,array,stmaryrd,xypic,tipa,slashed}
\usepackage{youngtab}%[vcentermath]
\usepackage[all]{xy}
\usepackage{lmodern}
\usepackage[utf8]{inputenc}
\usepackage{color}

\let\ssection=\section
\renewcommand{\section}{\setcounter{equation}{0}\ssection}

\newtheorem{thm}{Theorem}[section]
\newtheorem{lem}[thm]{Lemma}
\newtheorem{cor}[thm]{Corollary}
\newtheorem{prop}[thm]{Proposition}
\newtheorem{defi}[thm]{Definition}

\newtheorem{rmk}[thm]{Remark}

\def\a{\alpha}
\def\b{\beta}
\def\g{\gamma}
\def\Ga{\Gamma}
\def\d{\nu}			% \d et \bdel trop proches

\def\k{\kappa}
\def\l{\lambda}

\def\m{\mu}
\def\om{\omega}
\def\Om{\Omega}

\newcommand{\bbR}{\mathbb{R}}

\newcommand{\bbC}{\mathbb{C}}

\newcommand{\bbN}{\mathbb{N}}
\newcommand{\bbZ}{\mathbb{Z}}
\newcommand{\bbQ}{\mathbb{Q}}

\newcommand{\bi}{\mathsf{i}}

\newcommand{\cA}{{\mathcal{A}}}

\newcommand{\sfc}{\mathsf{c}}
\newcommand{\cC}{{\mathcal{C}}}
\newcommand{\tC}{C_{\sT}}
\newcommand{\Cd}{C_{\sS}}
\newcommand{\Clm}{C_{\sD}}

\newcommand{\ce}{\mathfrak{ce}}
\newcommand{\CE}{\mathrm{CE}}
\newcommand{\Cinfty}{{\mathcal{C}^{\infty}}}

\newcommand{\bbCl}{\mathrm{\bbC l}}
\newcommand{\cO}{{\mathcal{O}}}

\newcommand{\CO}{\mathrm{CO}}

\newcommand{\sD}{\mathsf{D}}
\newcommand{\D}{\mathcal{D}}

\newcommand{\Dirac}{\slashed{D}}
\newcommand{\Dlm}{\mathcal{D^{\l,\m}}}
\newcommand{\Dslm}{\mathsf{D^{\l,\m}}}

\newcommand{\cE}{\mathcal{E}}

\newcommand{\e}{\mathfrak{e}}

\newcommand{\rE}{\mathrm{E}}

\newcommand{\End}{\mathrm{End}}
\newcommand{\cF}{{\mathcal{F}}}

\newcommand{\fkg}{{\mathfrak{g}}}

\newcommand{\gr}{{\mathrm{gr}}}
\newcommand{\bgr}{{\mathrm{bigr}}}
\newcommand{\mg}{\text{\textg}}
\newcommand{\tg}{\tilde{\gamma}}

\newcommand{\Hom}{\mathrm{Hom}}
\newcommand{\HSQ}{\mathrm{QHS}}
\newcommand{\HSC}{\mathrm{CHS}}

\newcommand{\Id}{\mathrm{Id}}
\newcommand{\im}{\mathrm{im}}
\newcommand{\ccL}{\mathcal{L}}
\newcommand{\sL}{\mathsf{L}}
\newcommand{\LD}{\mathcal{L}_X^{\lambda,\mu}}

\newcommand{\Ld}{\mathsf{L}_X^\d}
\newcommand{\bbL}{\mathbb{L}}

\newcommand{\Lt}{\mathbb{L}_X^\d}

\newcommand{\fo}{\mathfrak{o}}

\newcommand{\cK}{{\mathcal{K}}}
\newcommand{\cKh}{{\mathcal{KH}}}

\newcommand{\Pol}{\mathrm{Pol}}

\newcommand{\cQ}{{\mathcal{Q}}}
\newcommand{\cQlm}{\mathcal{Q}^{\l,\m}}

\newcommand{\rL}{\mathrm{L}}
\newcommand{\cM}{{\mathcal{M}}}
\newcommand{\cN}{{\mathcal{N}}}
\newcommand{\rO}{\mathrm{O}}

\newcommand{\cS}{{\mathcal{S}}}

\newcommand{\sS}{\mathsf{S}}

\newcommand{\spo}{\mathfrak{spo}}
\newcommand{\spl}{\mathfrak{sp}}

\newcommand{\bbT}{\mathbb{T}}

\newcommand{\sT}{\mathsf{T}}

\newcommand{\ve}{\varepsilon}
\newcommand{\Vect}{\mathrm{Vect}}
\newcommand{\vol}{\mathrm{vol}}

\newcommand{\tp}{\tilde{p}}
\newcommand{\txi}{\tilde{\xi}}

\newcommand{\half}{\frac{1}{2}}

\newcommand{\fh}{\mathfrak{h}}

\newcommand{\fS}{\mathfrak{S}}

\newcommand{\bd}{\boldsymbol{d}}

\newcommand{\bdel}{\boldsymbol{\delta}}

\newcommand{\bDel}{\boldsymbol{Q}}%{\boldsymbol{\Delta}}
\begin{document}

\baselineskip=15pt

%%%%%%%%%%%%%%%%%%%%%%%%%%%%%%%%%%%%%%%%%%%%%%%%%%%%%%%%%%%%
%%%%%%%%%%%%%%%%%%%%%%%%% TITLE %%%%%%%%%%%%%%%%%%%%%%%%%%%%
\title{Conformally equivariant quantization for spinning particles}
%%%%%%%%%%%%%%%%%%%%%%%%%%%%%%%%%%%%%%%%%%%%%%%%%%%%%%%%%%%%
%%%%%%%%%%%%%%%%%%%%%%%%%%%%%%%%%%%%%%%%%%%%%%%%%%%%%%%%%%%%

\author{Jean-Philippe~Michel}
\address{University of Luxembourg, Campus Kirchberg, Mathematics Research Unit, 6, rue Richard Coudenhove-Kalergi, L-1359 Luxembourg City, Grand Duchy of Luxembourg}
\email{jean-philippe.michel@uni.lu}
\thanks{I thank the Luxembourgian NRF for support via the AFR grant PDR-09-063.
This research has been also partially funded by the Interuniversity Attraction 
Poles Program initiated by the Belgian Science Policy Office.}
\address{University of Li\`ege, 12 grande traverse, Sart-Tilman, B-4000 Li\`ege, Belgium}
\email{jean-philippe.michel@ulg.ac.be}

%\date{\today}

\keywords{Quantization, Conformal geometry, graded geometry, spinor differential operators, Dirac operator, higher symmetries, supercharges.}
%% Classification math�matique  (2000)
\subjclass[2010]{17B66, 17B81, 53A30, 53D55, 58J70, 70S10, 81S10}

\begin{abstract}
This work takes place over a conformally flat spin manifold $(M,\mg)$. We prove existence and uniqueness of the conformally equivariant quantization valued in spinor differential operators, and provide an explicit formula for it when restricted to first order operators. The Poisson algebra of symbols is realized as a space of functions on the supercotangent bundle $\cM=T^*M\oplus\Pi TM$, endowed with a symplectic form depending on the metric $\mg$. It admits two different actions of the conformal Lie algebra: one tensorial and one Hamiltonian. They are intertwined by the uniquely defined conformally equivariant superization, for which an explicit formula is given. This map allows us to classify all the conformal supercharges of the spinning particle in terms of conformal Killing tensors, which are symmetric, skew-symmetric or with mixed symmetry. Higher symmetries of the Dirac operator are obtained by quantization of the conformal supercharges.
\end{abstract}

\maketitle

%%% \tableofcontents

%%%%%%%%%%%%%%%%%%%%%%%%%%%%%%%%%%%%%%%%%%%%%%%%%%%%%%%%%%%%
%%%%%%%%%%%%%%%%%%%%%%%%%%%%%%%%%%%%%%%%%%%%%%%%%%%%%%%%%%%%
\section{Introduction}
%%%%%%%%%%%%%%%%%%%%%%%%%%%%%%%%%%%%%%%%%%%%%%%%%%%%%%%%%%%%
%%%%%%%%%%%%%%%%%%%%%%%%%%%%%%%%%%%%%%%%%%%%%%%%%%%%%%%%%%%%

Whereas there exists a pseudo-classical model for the spinning particle, due to Berezin and Marinov \cite{BMa77}, quantization is scarcely developed in that setting. We propose in this paper a natural extension of the conformally equivariant quantization for spinning particles, so that it is valued in spinor differential operators. We start this introduction with the concept of quantization and especially of conformally equivariant quantization of cotangent bundles as introduced in \cite{DLO99}. Then, we review the Hamiltonian formalism for pseudo-classical spinning particles and known quantizations for such systems. Finally we present our main results and detail the content of the paper.
%\\

\subsection{} The prototypical example of quantization is the one of cotangent bundles $T^*M$, 
endowed with their canonical symplectic structure. 
A quantization is then a linear map between a subalgebra of smooth functions on $T^*M$ 
and a subalgebra of linear operators acting on the Hilbert space $\rL^2(M)$
of square integrable functions on $M$.
Here, we focus on the algebra $\D(M)$ of scalar differential operators on $M$, 
which is filtered by the order of derivations.
Its associated graded algebra of symbols identifies with 
the algebra  $\Pol(T^*M)$  of fiberwise polynomial functions on $T^*M$, or equivalently 
with the algebra of symmetric tensors $\Ga(\cS TM)$.
As a result, we call quantization of $T^*M$ a linear map $\cQ:\Pol(T^*M)\rightarrow\D(M)$ 
which is the inverse of a full symbol map. 
Via the principal symbol maps, the algebras $\Pol(T^*M)\cong\Ga(\cS TM)$ inherit from $\D(M)$ 
a Poisson bracket and a Hamilonian action of vector fields. The bracket coincides with 
the canonical Poisson bracket on $\Pol(T^*M)$, % i.e. the Schouten bracket on $\Ga(\cS TM)$, 
while the action is the natural $\Vect(M)$-action on $\Pol(T^*M)\cong\Ga(\cS TM)$. 

The celebrated Weyl quantization of $T^*M$ is characterized, if $M=\bbR^n$, by its equivariance property under the action  on $T^*M$ of the symplectic affine Lie algebra $\spl(2n,\bbR)\ltimes\bbR^n$. Focusing rather on equivariance under a Lie algebra $\fkg$ acting by vector fields on the configuration manifold $M$ itself, Duval, Lecomte and Ovsienko have introduced the concept of {\it $\fkg$-equivariant quantization}. In particular, they prove its existence and uniqueness for $\fkg$ the projective \cite{LOv99} or the conformal Lie algebra \cite{DLO99}, acting locally on manifolds $M$ which are projectively or conformally flat respectively. It has been intensely developed since then, see e.g. \cite{DOv01,DEO04,Mic09,CSi09,Mic11a,Sil09} and references therein. 

A deep motivation to build quantizations is to obtain a correspondence between classical and quantum symmetries whenever there exists one, and an efficient way to measure the quantum anomalies otherwise. For a free massive scalar particle on $(M,\mg)$, this amounts to determining whether the first integrals of the geodesic flow of the metric $\mg$, given by symmetric Killing tensors, can be quantized into differential operators commuting with the Laplacian. This has been investigated for Killing $2$-tensors by Carter, 
using a minimal quantization procedure \cite{Car77}. There are no quantum anomalies in a number of examples, see e.g.\ \cite{DVa05}, in particular if $(M,\mg)$ is Ricci-flat. 
Over a conformally flat manifold, the situation is now well-understood thanks to the conformally equivariant quantization introduced above. Indeed, the latter establishes a bijection between the classical and quantum symmetries of a free massless particle, i.e., between the constants of motion along the null geodesic flow, given by conformal Killing symmetric tensors, and the higher symmetries of Laplacian \cite{Eas05,Mic11b}. Moreover, this correspondence extends to the massive case, for the symmetries given by differential operators of second order, and leads ultimately to new quantum integrable systems \cite{BEHRR11,DVa11}.
For additional informations on symmetries of Laplacian, we refer to \cite{MRS13,MSS14}.

\subsection{}To develop the spin counterpart of the preceding picture, we suppose that $(M,\mg)$ is a spin manifold of even dimension, with spinor bundle $S$. A quantum particle is now described by a square integrable section of $S$ and we focus on the algebra of spinor differential operators $\D(M,S)$. Weyl quantization \cite{Wid80} and conformally equivariant quantization \cite{CSi09} have been generalized to that setting, but with the usual algebra of symbols $\Pol(T^*M)\otimes_{\Cinfty(M)}\Ga(\bbCl(M,\mg))$ as source space. We rather follow Getzler \cite{Get83}, which in addition 
 uses the filtration of the Clifford bundle, to derive a supercommutative bigraded algebra of symbols from the usual one. This superalgebra identifies with the one of tensors $\Ga(\cS TM\otimes\Lambda T^*M)$ or equivalently with the algebra $\mathcal{O}(\cM)$ of fiberwise polynomial functions on the {\it supercotangent bundle} $\cM=T^*M\oplus\Pi TM$, where $\Pi$ denotes the reverse parity functor. This supermanifold is precisely the phase space introduced by Berezin and Marinov \cite{BMa77} to deal with pseudo-classical spinning particles on~$M=\bbR^n$. Besides, the author proved in \cite{Mic10a} the existence of a {\it Hamiltonian filtration} on $\D(M,S)$, assigning order $2$ to spinor covariant derivatives $\nabla_X$, with $X\in\Vect(M)$, and order $1$ to the Clifford elements $\g(\xi)$, with $\xi\in\Om^1(M)$. This filtration is compatible with the commutator of $\D(M,S)$ and induces a super Poisson bracket and a new gradation on $\cO(\cM)$. This grading stems from the identification of $\cM$ with the graded manifold $T^*[2]M\oplus T[1]M$, whereas the super Poisson bracket comes from the symplectic structure on $\cM$ induced by the metric $\mg$. The latter follows from a general construction, due to Rothstein in the super setting \cite{Rot90} and to Roytenberg in the graded setting \cite{Roy02}. In addition, the obtained symplectic structure on $\cM$ corresponds to the one coming from the Lagrangian of a free pseudo-classical spinning particle on $(M,\mg)$ \cite{Rav80,KCa11}, so that the graded Poisson algebra $\cO(\cM)$ is a classical counterpart to $\D(M,S)$. This is confirmed by the geometric quantization scheme, which associates the Hilbert space of square integrable spinors to the supercotangent bundle $\cM$ \cite{Mic10a}.

We name {\it quantizations} of $\cM$ the  linear maps $\cQ:\cO(\cM)\rightarrow\D(M,S)$ which are the inverse of a full symbol map for the Hamiltonian filtration of $\D(M,S)$. The Weyl quantization developed by Getzler \cite{Get83} and Voronov \cite{Vor90} is precisely of this type. In particular, it enables to quantize supercharges linear in momenta, built from Killing forms, into symmetries of the Dirac operator $\Dirac$ \cite{GRV93,Tan95}. However, Weyl quantization fails to quantize conformal supercharges, i.e. constants of motion of free massless spinning particles, into higher symmetries of the Dirac operator, even the simplest ones built from conformal Killing forms \cite{BCh97, BKr04}.

\subsection{} In this paper, we prove existence and uniqueness of a family $(\cQlm)_{\l,\m\in\bbR}$ of quantizations of the supercotangent bundle which are conformally equivariant, assuming that the base manifold $(M,\mg)$ is conformally flat. Such a map $\cQlm$ allows, for right parameters $(\l,\m)$, to quantize all the conformal supercharges into higher symmetries of the Dirac operator, i.e.\ into operators $D_1\in\D(M,S)$ such that $\Dirac D_1=D_2\Dirac$ for some $D_2\in\D(M,S)$. In particular, we recover the symmetries of $\Dirac$ given by first order differential operators, classified previously in \cite{BKr04}. 
This is a first step toward the study of the algebra structure of the higher symmetries of the Dirac operator, which should be of interest for higher spin field theories. 
Future applications to separability of the Dirac equation are also expected. The determination of second order symmetries of the Dirac operator have already been investigated over curved $4$-manifolds \cite{ABB14}.
By the way, we also build a family $(\fS^\d)_{\d\in\bbR}$ of so-called conformally equivariant {\it superizations} and provide explicit formul{\ae} for them. For $\d=0$, the map $\fS^\d$ establishes a correspondence between conformal Killing tensors with mixed symmetry, lying in $\Ga(\cS TM\otimes\Lambda T^*M)$, and conformal supercharges. 
This classification of conformal supercharges extends previous results by physicists \cite{GRV93,Tan95} and may prove useful for integrability of spinning particle motion, like in \cite{KCa11}.

\subsection{}We detail now the content of the present paper.

In Sect.\ $2$, we review some results of \cite{Mic10a}.
In particular, we recall the gradation and the filtration on the source and 
target spaces of the conformally equivariant quantization of $\cM$  
 and provide the actions of 
the Lie algebra $\fkg$ of conformal Killing vector fields on them. 
All through this section, we assume that $(M,\mg)$ is conformally flat 
and of signature $(p,q)$ so that $\fkg\cong\fo(p+1,q+1)$. Its action 
on the spinor bundle $S$ is given by the Kosmann's Lie derivative 
of spinors \cite{Kos72}, which admits a one parameter deformation 
obtained geometrically by tensoring $S$ with a real power of the determinant bundle, 
i.e. with a density bundle. We denote by $(\Dslm)_{\l,\m\in\bbR}$ 
the family of $\fkg$-modules defined by the induced deformed adjoint 
actions of $\fkg$ on $\D(M,S)$. According to \cite{Mic10a}, 
the action of $\fkg$ on $\Dslm$ yields to two different 
actions on the symbol algebra $\cO(\cM)$. The first one is Hamiltonian 
and induced by the principal symbol maps associated to the Hamiltonian 
filtration of $\D(M,S)$, it depends only on the shift $\d=\m-\l$ and leads 
to the $\fkg$-module denoted $\sS^\d$. It should be understood as the 
module of classical observables. The latter $\fkg$-action does not preserve 
the bigradation of $\cO(\cM)\cong\Ga(\cS TM\otimes\Lambda T^*M)$, but 
via the {\it principal tensorial symbol maps}, it induces a second $\fkg$-action 
on $\cO(\cM)$ which does. We denote by $\sT^\d$ this extra $\fkg$-module, 
which identifies with the tensor module $\Ga(\cS TM\otimes\Lambda T^*M)$ up 
to a twist by appropriate density bundles. Then, we define the conformally 
equivariant superization $\fS^\d:\sT^\d\rightarrow\sS^\d$, inverse to a full 
tensorial symbol map, and the conformally equivariant quantization 
$\cQlm:\sS^\d\rightarrow\Dslm$, inverse to a full Hamiltonian symbol map. 
The name superization comes from the inclusion $\Pol(T^*M)\subset\sT^0$ as 
$\Vect(M)$-modules, so that each non-spinning classical observable admits 
a spinning analog via $\fS^0$.

In Sect.\ $3$, we introduce the building blocks of the maps $\fS^\d$ and $\cQlm$, which are the isometric invariant differential operators acting on $\sT^\d$. They form an algebra which is easily determined thanks to Weyl's theory of invariants \cite{Wey97}. It is generated by $13$ operators, most of which are well-known, e.g., divergence, gradient, de Rham differential. Some variants of this algebra naturally appear in the context of Howe dual pairs \cite{How89,LHo10} and were investigated over constant curvature manifolds \cite{HWa05,HWa07,HWa08}. Except the computation of  the Casimir operators of the three $\fkg$-modules $\sT^\d$, $\sS^\d$ and $\Dslm$, the material in this section is not new. However, the presentation, which uses the identification of differential operators on the tensor space $\sT^\d$ with scalar differential operators on $\cM$, is original and proves to be convenient. 

In Sect.\ $4$, we prove our main theorem: existence and uniqueness of the conformally equivariant superization $\fS^\d:\sT^\d\rightarrow\sS^\d$ and  quantization $\cQlm:\sS^{\d}\rightarrow\Dslm$, with $\d=\m-\l\in\bbR$, except for a discrete subset of {\it critical weights} $\d$. So far, only the existence and uniqueness of the composition of these two maps was known, as a consequence of the general work \cite{CSi09} on equivariant quantizations. Our proof follows the one of existence and uniqueness of conformally equivariant quantization  in the scalar setting \cite{DLO99}. The main point is the diagonalization of the Casimir operator of $\sT^\d$, which proves to be equivalent to the harmonic decomposition of the polynomial superalgebra $\cS\bbR^n\otimes\Lambda\bbR^n$. After deriving it, we discover this was the purpose of \cite{Hom01}. We nevertheless include our proof for completeness. 
Moreover, we determine the critical weights $\d$ for $\fS^\d$ (resp.\ $\cQlm$). According to \cite{Mic11a}, they coincide with the weights $\d$ for which there exists a conformally invariant operator on $\sT^\d$ (resp.~$\sS^\d$), which strictly lowers the degree. The classification of such operators can be deduced from \cite{BCo85a,BCo85b}, we get it here from the more basic Weyl's theory of invariants. It allows us to show that critical weights are positive, so that existence and uniqueness of $\fS^\d$ and $\cQlm$ holds in the most usual case $\d=\m-\l=0$.

In Sect.\ $5$, we compute explicit formul{\ae} for $\fS^\d$ and $\cQlm$.
Following same method as in \cite{DOv01}, we get a local formula for the conformally equivariant superization. Then, we obtain such formula for the conformally equivariant quantization also, when it is valued in first order operators. After that, we derive covariant expressions  for both maps. Contrary to the scalar case, none of these maps is conformally invariant, i.e., depends only on the conformal class of $\mg$. This is due to the fact that the symplectic form on $\cM$ depends on $\mg$ and therefore the Hamiltonian $\fkg$-action on $\sS^\d$ also does. Only their composition $\cQlm\circ\fS^\d$, with $\d=\m-\l$, is a conformally invariant map. 
  
In Sect.\ $6$, we classify the symmetries of a free massless spinning particle over a conformally flat manifold. 
In the pseudo-classical case, they are called conformal supercharges and arise as the conformally equivariant superization of conformal Killing hook tensors. The latter are given by the traceless component 
of the tensor product of symmetric and skew-symmetric conformal Killing tensors. 
After quantization, the conformal supercharges correspond to the higher symmetries of the Dirac operator. 
Both classifications of conformal supercharges and of higher symmetries of the Dirac operator are new. 
They generalize the description of first order symmetries of the pseudo-classical \cite{GRV93,Tan95} and
quantum \cite{BKr04} spinning particle.
For extension of second order symmetries to the curved case and applications to integrability, we refer to
\cite{KCa11,CKK11,ABB14}.

We include an appendix, where we collect the needed basic computations in the algebra of isometric invariant differential operators acting on the tensor space $\sT^\d$.
\\

In the paper, we work over a pseudo-Riemannian manifold $(M,\mg)$, with $\mg$ a metric of signature $(p,q)$. That manifold is supposed to be spin and of even dimension $n=p+q$. We freely use Einstein's summation convention. The tensor products of $\Cinfty(M)$-modules, taken over the algebra $\Cinfty(M)$, are denoted by $\otimes_\Cinfty$.

%%%%%%%%%%%%%%%%%%%%%%%%%%%%%%%%%%%%%%%%%%%%%%%%%%%%%%%%%%%%
%%%%%%%%%%%%%%%%%%%%%%%%%%%%%%%%%%%%%%%%%%%%%%%%%%%%%%%%%%%%
\section{Spinor Differential Operators and their Symbols}
%%%%%%%%%%%%%%%%%%%%%%%%%%%%%%%%%%%%%%%%%%%%%%%%%%%%%%%%%%%%
%%%%%%%%%%%%%%%%%%%%%%%%%%%%%%%%%%%%%%%%%%%%%%%%%%%%%%%%%%%%

In order to keep a self-contained exposition, this section presents a recollection of our previous work \cite{Mic10a}. 
Namely, we introduce a family of $\fo(p+1,q+1)$-module structures on the space of spinor differential operators, denoted by $(\Dslm)_{\l,\m\in\bbR}$, and on its two spaces of symbols, denoted by $(\sS_\d)_{\d\in\bbR}$ and $(\sT_\d)_{\d\in\bbR}$. 
Then, we compare these three family of modules. 
%This is done in a closed spirit to \cite{DLO99}, the seminal paper on conformally equivariant quantization of cotangent bundles.

%%%%%%%%%%%%%%%%%%%%%%%%%%%%%%%%%%%%%%%%%%%%%%%%%%%%%%%%%%%%
\subsection{The supercotangent bundle}
%%%%%%%%%%%%%%%%%%%%%%%%%%%%%%%%%%%%%%%%%%%%%%%%%%%%%%%%%%%%   

Let $E$ be a vector bundle over the smooth manifold~$M$ and $d$ an integer. We denote by $E[d]$ the $\bbN$-graded manifold with base $M$ and structural sheaf~$\cO(E[d])$, which identifies with the graded sheaf of complexified sections of $\cS E^*$ if $d$ is even and of $\Lambda E^*$ if $d$ is odd. The sections of $\cS^k E^*$, or $\Lambda^k E^*$, receive a degree $kd$. Such a formalism allows to encode symmetric contravariant tensors, i.e.\ sections of the bundle $\cS TM$, as functions on $T^*[2]M$, the order of tensors being twice the degree of corresponding functions. Analogously, the algebra of complexified differential forms $\Om^\bbC(M)$ is viewed as the algebra of functions on $T[1]M$, the degree of forms equating the one of functions. 

Both types of tensors are encompassed into the algebra of functions of the supercotangent bundle of $M$, defined as the Whitney sum  $\cM=T^*[2]M\oplus T[1]M$.
The gradation of its structural sheaf is given by the following subspaces, for $\ell\in\bbN$,
\begin{equation}\label{Olgraduation}
\cO_{[\ell]}(\cM)=\bigoplus_{2k+\k=\ell}\cO_{k,\k}(\cM)\quad \text{and} \quad\cO_{k,\k}(\cM)=\Pol_k(T^*M)\otimes_{\Cinfty}\Om^\bbC_\k(M), 
\end{equation}
where $\Pol_k(T^*M)\cong\Ga(\cS^k TM)$ is the space of functions on $T^*M$ of degree $k$ in the fiber variables and $\Om^\bbC_\k(M)$ is the space of complex differential forms of degree $\k$. Starting from a local coordinate system $(x^i)$ on $M$, we can build a natural one $(x^i,p_i,\xi^i)$ on $\cM$, with $p_i$ identifying to the partial derivative $\partial_i:=\partial/\partial x^i$ and $\xi^i$ to the differential $1$-form $dx^i$.

Rothstein has classified the even symplectic structures on supermanifolds in \cite{Rot90}. Accordingly, a symplectic structure on $\cM$ is equivalent to the three following piece of data: a symplectic form on $T^*M$, a metric on the vector bundle $TM$, i.e.\ a metric on $M$, and a compatible connection. This leads to the following proposition.

\begin{prop}\cite{Mic10a}
Let $(M,\mg)$ be a pseudo-Riemannian manifold and $\hbar\in\bbR$. Its supercotangent bundle $\cM=T^*[2]M\oplus T[1]M$ admits an exact symplectic structure $\om=\bd \a$, whose potential $1$-form reads in local natural coordinates as
\begin{equation}\label{alpha}
\a =p_i\bd x^i+\frac{\hbar}{2\bi}\mg_{ij}\xi^i \bd^\nabla \xi^j,
\end{equation}
with $\bd^\nabla$ the covariant differential associated to the Levi-Civita connection.
\end{prop} 

The introduction of the factor $\hbar/\bi$ into the symplectic form $\om$ deserves some explanations. 
First, from a physical point of view, a symplectic form should have the dimension of an action. 
The odd variables $\xi$ being dimensionless, 
we have to insert a constant with the dimension of an action in front of the $\xi$-term. 
We choose~$\hbar$ since our aim is quantization. To deal with pseudo-classical spinning particles \cite{BMa77,Rav80,Mic09,KCa11}, $\hbar$ should be replaced by a characteristic action of the studied system. Second, the symplectic form should be real. For quantization purposes, we define the real structure on $\cM$ by the involution anti-automorphism 
\begin{equation}\label{conjugation}
\bar{\cdot}:\cO(\cM)\rightarrow\cO(\cM)
\end{equation}
equal to the identity on coordinates. In particular, we get $\overline{\xi^i\xi^j}=\bar{\xi^j}\bar{\xi^i}=-\xi^i\xi^j$ and $\a$, $\om$ are then real.

Since the symplectic form $\om$ is of degree $2$, $(\cM,\om)$ fits also into the classification of graded symplectic manifold of degree $2$, performed in \cite{Roy02}  by Roytenberg. In addition, the Poisson bracket associated to $\om$ lowers the degree by $2$. 
Thanks to Leibniz property, it is completely determined by the following equalities 
$$
\{X,Y\}=\nabla_{[X,Y]}+\frac{\hbar}{2\bi}R^k_{~lij}\xi_k\xi^lX^iY^j, \quad \{X,\xi\}=\nabla_X\xi, 
\quad \{\xi,\xi'\}=-\frac{\bi}{\hbar}\mg^{-1}(\xi,\xi'),
$$
where $X,Y\in\Pol_1(T^*M)$ identify to vector fields, $\xi,\xi'\in\Om^1(M)$, $\mg^{-1}$ is the metric induced by $\mg$ on $T^*M$, $\nabla$ is the Levi-Civita connection and $(R^k_{~l})_{ij}=[\nabla_i,\nabla_j]$ are the components of its Riemann tensor.

%%%%%%%%%%%%%%%%%%%%%%%%%%%%%%%%%%%%%%%%%%%%%%%%%%%%%%%%%%%%
\subsection{Spinor differential operators}
%%%%%%%%%%%%%%%%%%%%%%%%%%%%%%%%%%%%%%%%%%%%%%%%%%%%%%%%%%%%   

Let $(M,\mg)$ be a pseudo-Riemannian spin manifold of even dimension $n$. Its spinor bundle $S$ satisfies $\End S\cong\bbCl(M,\mg)$, where $\bbCl(M,\mg)$ is the complex Clifford bundle of $(M,\mg)$. We introduce the Clifford quantization map
\begin{equation}\label{gamma}
\g:\Om^\bbC(M)\rightarrow\Ga(\bbCl(M,\mg)),
\end{equation} 
which satisfies the Clifford relation $\g(\xi)\g(\xi')+\g(\xi')\g(\xi)=-2\mg^{-1}(\xi,\xi')$ for any $\xi,\xi'\in\Om^\bbC_1(M)$. 
The algebra $\D(M,S)$ of differential operators acting on sections of $S$ is an algebra filtered by the order of derivations, over the subalgebra $\Ga(\End S)\cong\Ga(\bbCl(M,\mg))$ of zeroth order operators.
It inherits a $\bbZ_2$-grading from $\Ga(\bbCl(M,\mg))$, so that 
the commutator of homogeneous elements is given by $[A,B]=AB-(-1)^{|A||B|}BA$. 
The filtration of $\bbCl(M,\mg)$ and a spinor connection on $S$ allow to define 
a naive bifiltration of $\D(M,S)$ by the following subspaces, indexed by $k\in\bbN$ and $\k\leq n$,
$$
\D^\nabla_{k,\k}(M,S)=\text{span}\{\g(\xi_1)\cdots\g(\xi_\m)\nabla_{X_1}\cdots\nabla_{X_m}\,|\,m\leq k \text{ and }\m\leq \k\}.
$$
Here, $\xi_i$ pertains to $\Om_1^\bbC(M)$ and $\nabla_{X_i}$ is the spinor covariant derivative along the vector field $X_i$. 
Such a bifiltration behaves badly with respect to the commutator, indeed $[\nabla_X,\nabla_Y]$ is of degree $2$ in $\g(\xi)$. Moreover, it depends on the choice of connection, as $\nabla_X=X^i\partial_i+X^i\Ga_{ij}^k\g(\xi^j)\g(\xi_k)$. 
These difficulties are overcome by the following definition, introduced in \cite{Mic10a},
\begin{equation}\label{filtration:Dkk}
\D_{k,\k}(M,S)=\text{span}\{\g(\xi_1)\cdots\g(\xi_\m)\nabla_{X_1}\cdots\nabla_{X_m}\,|\,m\leq k \text{ and } 2m+\m\leq 2k+\k\}.
\end{equation}
For such a bifiltration, if $A$ and $B$ are two operators of orders $(k,\k)$ and $(k',\k')$, then their product $AB$ is of order $(k+k',\k+\k')$ and their commutator $[A,B]$ splits into two terms of orders $(k+k'-1,\k+\k')$ and $(k+k',\k+\k'-2)$.
Thus, the associated bigraded algebra
$$
\bgr\D(M,S):=\bigoplus_ {(k,\k)\in\bbN\times\llbracket 0,n\rrbracket}\D_{k,\k}(M,S)\,/\,\left(\D_{k-1,\k}(M,S)+\D_{k,\k-1}(M,S)\right)
$$
is a graded-commutative algebra.
The above bifiltration  is built from two filtrations on $\D(M,S)$, the usual one by the order of derivations 
and a new one given by the subspaces
\begin{equation}\label{HamFil}
\D_{[\ell]}(M,S)=\bigcup_{2k+\k=\ell}\D_{k,\k}(M,S),
\end{equation}
with $\ell\in\bbN$. In particular, it assigns the same order to $\nabla_X$ and $\g(\xi)\g(\xi')$. 
This new filtration is compatible with the commutator so that the associated graded algebra, 
$$
\gr\D(M,S):=\bigoplus_{\ell\in\bbN}\D_{[\ell]}(M,S)/\D_{[\ell-1]}(M,S),
$$ 
is supercommutative and endowed with a Poisson bracket of degree $-2$. 
In consequence, it is referred thereafter as the Hamiltonian filtration. 

\begin{prop}\cite{Mic10a}
We get the following isomorphisms 
\begin{eqnarray*}
\bgr\D(M,S)&\cong & \bigoplus_{(k,\k)\in\bbN\times\llbracket 0,n\rrbracket}\Ga(\cS^k TM\otimes\Lambda^\k T^*M),\\ %\cO(\cM) \cong  \Ga(\cS TM\otimes\Lambda T^*M),\\
\gr\D(M,S)&\cong & \bigoplus_{\ell\in\bbN} \cO_{[\ell]}(\cM),
\end{eqnarray*}
of respectively bigraded algebras and graded Poisson algebras. 
\end{prop}   

Accordingly, we introduce the Hamiltonian principal symbol maps defined by the following projections,
\begin{equation}\label{symbolDS}
%\tau_{k,\k}:\D_{k,\k}(M,S)&\rightarrow &\Ga(\cS^k TM\otimes\Lambda^\k T^*M),\\
\sigma_\ell:\D_{[\ell]}(M,S)\rightarrow \cO_{[\ell]}(\cM),
\end{equation}
which are normalized by $\sigma_2(\nabla_{\partial_i})=p_i$ and $\sigma_1\left(\frac{\g(\xi^i)}{\sqrt{2}}\right)=\left(\frac{\hbar}{\bi}\right)^{1/2}\xi^i$. 
They satisfy the usual relations: 
\begin{equation}\label{PcpalSymbol}
\sigma_{\ell}(A)\sigma_{\ell'}(B)=\sigma_{\ell+\ell'}(AB) \quad \text{and} \quad\{\sigma_\ell(A),\sigma_{\ell'}(B)\}=\sigma_{\ell+\ell'-2}([A,B]),
\end{equation} 
for all differential operators $A,B$ of Hamiltonian orders $\ell$ and $\ell'$.

We denote by $|\Lambda|^\l:=|\Lambda^\mathrm{top}T^*M|^{\otimes\l}$ the line bundle of $\l$-densities,  
with $\l\in\bbR$.
Note that a metric provides a canonical trivialization 
of $|\Lambda|^\l$ via $|\vol_\mg|^\l$, with $\vol_\mg$ the volume form defined by $\mg$.
Since $1$-densities are the natural objects for integration, 
a pseudo-Hermitian pairing $\langle\cdot,\cdot\rangle_S$ on $S$ yields a scalar product between 
compactly supported sections  
$\phi\in\Ga(S\otimes|\Lambda|^\l)$ and $\psi\in\Ga(S\otimes|\Lambda|^{1-\l})$, given by $\int_M\langle\phi,\psi\rangle_S$.
The definition of the adjoint $D^*$ of an operator $D:\Ga(S\otimes|\Lambda|^\l)\rightarrow\Ga(S\otimes|\Lambda|^{1-\l})$ 
follows,
\begin{equation}\label{adjunction}
\int_M\langle\phi,D\psi\rangle_S=\int_M\langle D^*\phi,\psi\rangle_S
\end{equation}
for all compactly supported sections $\phi,\psi\in\Ga(S\otimes|\Lambda|^\l)$. 
In the sequel, we choose $\langle\cdot,\cdot\rangle_S$ such that 
$\g(\overline{\eta})=\g(\eta)^*$, for all $\eta\in\Om^\bbC(M)$.

%%%%%%%%%%%%%%%%%%%%%%%%%%%%%%%%%%%%%%%%%%%%%%%%%
\subsection{The conformal Lie algebra $\fkg$}
%%%%%%%%%%%%%%%%%%%%%%%%%%%%%%%%%%%%%%%%%%%%%%%%%

Let $(M,\mg)$ be a pseudo-Riemannian manifold of signature $(p,q)$. A conformal Killing vector field on $M$ is a vector field $X$ whose Lie derivative action preserves the direction of the metric, i.e.\ $L_X\mg$ is proportional to $\mg$.
By definition, $(M,\mg)$ is conformally flat if $\mg=F\eta$ locally, where $F$ is a positive function and $\eta$ is the flat metric of signature $(p,q)$. Then, the Lie algebra $\fkg$ of local conformal Killing vector fields of $(M,\mg)$ is of maximal dimension and identifies with $\fo(p+1,q+1)$. 
On $(\bbR^{p,q},F\eta)$, the Lie algebra $\fkg\cong\fo(p+1,q+1)$ of conformal Killing vector fields is given explicitly, in terms of its generators, by
\begin{eqnarray} \nonumber
X_i&=&\partial_i,\\ \nonumber
X_{ij}&=& x_i\partial_j- x_j\partial_i,\\ \label{AlgLieConf}
X_0&=& x^i\partial_i,\\ \nonumber
\bar{X}_i&=&x_jx^j\partial_i - 2x_ix^j\partial_j,
\end{eqnarray}
for $i,j=1,\ldots,n$, the indices being lowered using $\eta$. The latter vector fields provide the local realization of $\fkg$ on $(M,\mg)$, endowed with a $|1|$-gradation $\fkg:=\fkg_{-1}\oplus\fkg_0\oplus\fkg_1$ given by the polynomial degree of the coefficients. With respect to the metric $\eta$, $\fkg_{-1}$ is the Lie subalgebra of translations, $\fkg_0$~splits into the Lie subalgebra $\fo(p,q)$ of rotations and its center, generated by the Euler vector field~$X_0$, and~$\fkg_1$~is~the subspace of conformal inversions. We denote the Lie algebra of isometries by $\e(p,q):=\fo(p,q)\ltimes\fkg_{-1}$, and the one of similitudes by $\ce(p,q):=\bbR X_0\ltimes\e(p,q)$.

%%%%%%%%%%%%%%%%%%%%%%%%%%%%%%%%%%%%%%%%%%%%%%%%%
\subsection{The $\fkg$-module of tensorial symbols $\sT^\d$}
%%%%%%%%%%%%%%%%%%%%%%%%%%%%%%%%%%%%%%%%%%%%%%%%%

Any vector field $X$ on~$M$ admits a canonical lift to the linear frame bundle via its first order jet $(\partial_iX^j)$. This leads to an action of $\Vect(M)$ by Lie derivatives on any associated bundle, e.g. $\ell^\l_X=X^i\partial_i+\l\partial_i X^i$ on $\l$-densities.
The tensor bundle $\cS TM\otimes\Lambda T^*M$ is also associated to the principal frame bundle, therefore it admits an action of $\Vect(M)$ by Lie derivatives. Under the identification of such tensors with functions on the supercotangent bundle $\cM$, the Lie derivative corresponds to the following lift of $\Vect(M)$ to $\Vect(\cM)$
\begin{equation}\label{Lift}
\Vect(M)\ni X\mapsto \hat{X}:=X^i\partial_i+(\partial_iX^j)(\xi^i\partial_{\xi^j}-p_j\partial_{p_i}),
\end{equation}
where $(x^i,p_i,\xi^i)$ denotes a local natural coordinate system on $\cM$ and $(\partial_i,\partial_{p_i},\partial_{\xi^i})$ are the corresponding partial derivatives. Clearly, this $\Vect(M)$-action preserves the bigradation of the tensor algebra $\Ga(\cS TM\otimes\Lambda T^*M)$.
\begin{defi}\label{defi:T}
Let $\d\in\bbR$. We set 
$\bbT^\d:=\bigoplus_{\k\leq n}\cS TM\otimes\Lambda^\k T^*_\bbC M\otimes|\Lambda|^{\d-\frac{\k}{n}}$. 
The bigraded module of tensorial symbols is the bigraded space of sections $\sT^\d:=\Ga(\bbT^\d)$ 
endowed with the natural $\Vect(M)$-action, given on the $\k$-component by
$$
\bbL_X^\d=\hat{X}\otimes\Id+\Id\otimes\ell^{\d-\k/n}_X.
$$
%The elements of $\sT^\d$ will be referred to as tensorial symbols.
\end{defi}

The usual space of symbol of $\D(M,S)$ is its graded algebra for the filtration by the order of derivations. It identifies with $\Pol(T^*M) \otimes_{\Cinfty}\Ga(\bbCl(M,\mg))$ and is related to $\sT^\d$ as follows.

\begin{prop}\label{prop:symb-symbT}
Let $(M,\mg)$ be a pseudo-Riemannian manifold and $\Sigma:=\xi^i\partial_{\xi^i}$ be the Euler vector field of $T[1]M$. 
The following map
\begin{equation}\label{Invgamma}
 \Id\otimes|\vol_\mg|^{\frac{\Sigma}{n}}\g:\sT^\d\rightarrow\Ga(\cS TM\otimes|\Lambda|^\d)\otimes_{\Cinfty}\Ga(\bbCl(M,\mg)),
\end{equation}
is a linear isomorphism which depends only on the conformal class of $\mg$. 
\end{prop}

%%%%%%%%%%%%%%%%%%%%%%%%%%%%%%%%%%%%%%%%%%%%%%%%%
\subsection{The $\fkg$-module of Hamiltonian symbols $\sS^\d$}\label{Sec:S}
%%%%%%%%%%%%%%%%%%%%%%%%%%%%%%%%%%%%%%%%%%%%%%%%%

In contradistinction with the cotangent bundle case, the natural Lift \eqref{Lift} of $\Vect(M)$ by Lie derivatives 
does not lead to Hamiltonian vector fields on $\cM$. Besides, the preservation of the potential $1$-form $\a$, 
see \eqref{alpha}, is not a strong enough condition to determine a unique lift of $X\in\Vect(M)$ to $\cM$. 
In \cite{Mic10a}, we ask in addition for preservation of the direction of the $1$-form $\b=\mg_{ij}\xi^idx^j$. 
Both conditions can be satisfied only for vector fields $X\in\fkg$, 
and fix a unique lift 
\begin{equation}\label{Ham-Lift}
\fkg\ni X\mapsto \tilde{X}:=\hat{X}+\frac{\Sigma}{n}+\frac{\hbar}{2\bi}\left(R^k_{~lij}\xi_k\xi^lX^j+\nabla_i(\partial_{[l}X_{k]})\xi^k\xi^l\right)\partial_{p_i},
\end{equation}
the brackets denoting skew-symmetrization. The action of $\tilde{X}$ clearly preserves the gradation $\cO(\cM)=\bigoplus_{\ell\in\bbN}\cO_{[\ell]}(\cM)$ defined in \eqref{Olgraduation}. 

\begin{defi}
Let $\d\in\bbR$. 
%We denote by $\sS^\d$ the graded $\fkg$-module $\cO(\cM)\otimes_\Cinfty\Ga(|\Lambda|^\d)$ of Hamiltonian symbols, 
The graded $\fkg$-module  of Hamiltonian symbols is the space $\sS^\d:=\cO(\cM)\otimes_\Cinfty\Ga(|\Lambda|^\d)$, 
%where the action of $X\in\fkg$ is given by
endowed with the following $\fkg$-action% of $X\in\fkg$
$$
\sL_X^\d=\tilde{X}\otimes\Id+\Id\otimes\ell^\d_X,
$$
and the gradation is given by $\sS^\d=\bigoplus_{\ell\in\bbN}\sS^\d_{[\ell]}$, with
$\sS^\d_{[\ell]}:=\cO_{[\ell]}(\cM)\otimes_\Cinfty\Ga(|\Lambda|^\d)$ (see \eqref{Olgraduation}).
\end{defi}

Each of the $\fkg$-submodule $\sS^\d_{[\ell]}$ admits a filtration by the p-degree, i.e. by the spaces
$$
\sS^\d_{k,\k}:=\bigoplus_{j\in\bbN}
\cO_{k-j,\k+2j}(\cM)\otimes_\Cinfty\Ga(|\Lambda|^\d),
$$
with $2k+\k=\ell$. Here, by convention, $\cO_{k,\k}(\cM)=\{0\}$ if $k<0$ or $\k<0$. 
Explicitly, if $\ell=2\ell_0$ is even, the filtration of $\sS^\d_{[\ell]}$ takes the form
\begin{equation}\label{OrderFil}
\sS^\d_{\ell_{0}-\frac{n}{2},n}\subset\sS^\d_{\ell_0-\frac{n-2}{2},n-2}\subset\cdots\subset\sS^\d_{\ell_0,0}=\sS^\d_{[\ell]}.
\end{equation}
In view of \eqref{Ham-Lift}, the spaces $\sS^\d_{k,\k}$ 
are preserved by the $\fkg$-action and the associated  graded $\fkg$-module 
$\gr\sS^\d_{[\ell]}$ is isomorphic to $\bigoplus_{2k+\k=\ell}\sT^\d_{k,\k}$.
Hence, $\sT^\d$ is isomorphic to the bigraded module $\bgr\sS^\d$, associated with the bifiltration
defined by the spaces $\sS^\d_{k,\k}$.
The principal tensorial symbol maps, defined by projections to the highest $p$-degree component, read then as
\begin{equation}\label{symbolST}
\ve_{k,\k}:\sS^\d_{k,\k}\rightarrow\sT^\d_{k,\k}.
\end{equation}

%%%%%%%%%%%%%%%%%%%%%%%%%%%%%%%%%%%%%%%%%%%%%%%%%
\subsection{The $\fkg$-module of spinor differential operators $\Dslm$}
%%%%%%%%%%%%%%%%%%%%%%%%%%%%%%%%%%%%%%%%%%%%%%%%%

Contrary to tensor bundles, the spinor bundle does not admit a canonical Lie derivative. 
Following Kosmann \cite{Kos72}, we set $L_X:=\nabla_X+\g(\bd X^\flat)/2$ 
with $X^\flat$ the $1$-form deduced from $X$ by the metric $\mg$. 
This formula gives a representation of Lie algebras only if restricted to conformal Killing vector fields $X\in\fkg$. 
The induced Lie derivative on $\Ga(S\otimes|\Lambda|^\l)$ is given by $L^\l_X:=L_X\otimes\Id+\Id\otimes\ell_X^\l$.   

\begin{defi}
Let $\l,\m\in\bbR$. The $\fkg$-module $\Dslm:=\D(M;S\otimes|\Lambda|^\l,S\otimes|\Lambda|^\m)$  
is the space of spinor differential operators where
the action of $X\in\fkg$ on an element $A\in\Dslm$ is given by the following Lie derivative
$$
\LD A=L^\m_X A-A L^\l_X.
$$
\end{defi}

The bifiltration of $\D(M,S)$, defined by the subspaces $\D_{k,\k}(M,S)$ (see \eqref{filtration:Dkk}),
endows $\Dslm$ with a bifiltration which is preserved by the above $\fkg$-action $\ccL^{\l,\m}$. 
Hence it induces $\fkg$-modules structure on $\bgr\Dslm$, $\gr\Dslm$ and on the usual symbol space $\Ga(\cS TM\otimes|\Lambda|^{\m-\l})\otimes_\Cinfty\Ga(\bbCl(M,\mg))$. We compare them below.

\begin{prop}\label{IsoGrDST}
For $\d=\m-\l$, we have the following isomorphisms of (bi)graded $\fkg$-modules: 
 $\gr\Dslm\cong\sS^\d$ and $\bgr\Dslm\cong\bgr\sS^\d\cong \sT^\d$. 
Moreover, the Map \eqref{Invgamma} turns into a $\fkg$-module isomorphism between $\sT^\d$ 
and $\Ga(\cS TM\otimes|\Lambda|^{\d})\otimes_\Cinfty\Ga(\bbCl(M,\mg))$. 
\end{prop}

%%%%%%%%%%%%%%%%%%%%%%%%
\subsection{Explicit formul{\ae} for actions of $\fkg$}
%%%%%%%%%%%%%%%%%%%%%%%%

We restrict here to the local model $(\bbR^{p,q},F\eta)$ of a conformally flat manifold. 
All the density bundles can be trivialized by powers of the $1$-density
$|\vol_\eta|=|dx^1\wedge\cdots\wedge dx^n|$. Moreover,
the supercotangent bundle admits natural cartesian coordinates $(x^i,p_i,\xi^i)$ which transform tensorially, $p_i$ identifying to $\partial_i$ and $\xi^i$ to $dx^i$. The potential $1$-form $\a$ of $\cM$ is written in terms of such coordinates in Eq.\ \eqref{alpha}. 
The supercotangent bundle carries also Darboux coordinates $(x^i,\tp_i,\txi^i)$, 
such that $\a=\tp_i\bd x^i+\frac{\hbar}{2\bi}\eta_{ij}\txi^i\bd \txi^j$ . 
The odd coordinates $(\txi^i)$ form an orthonormal frame of $T^*M$ 
and we have $\tp_i=p_i+(\hbar/2\bi)\Ga_{ij}^{k}\xi^j\xi_k$. 
In terms of spinor differential operators, the natural coordinates satisfy 
$p_i=\sigma_2(\nabla_i)$ and $\g(\xi^i)\g(\xi^j)+\g(\xi^j)\g(\xi^i)=-2F^{-1}\eta^{ij}$, 
whereas the Darboux coordinates satisfy $\tp_i=\sigma_2(\partial_i)$ and $\g(\txi^i)\g(\txi^j)+\g(\txi^j)\g(\txi^i)=-2\eta^{ij}$. 
Using both coordinate systems, we introduce a linear isomorphism, 
\begin{eqnarray}\label{Tenseur_Poly_Moments}
\cF:\sT^\d&\longrightarrow& \sS^\d\\ \nonumber
P^{i_1\dots i_k}_{j_1\ldots j_\k}(x)\, \xi^{j_1}\ldots \xi^{j_\k}  p_{i_1}\ldots p_{i_k} &\longmapsto&
|\vol_\eta|^{\frac{\k}{n}}P^{i_1\dots i_k}_{j_1\ldots j_\k}(x)\,  \txi^{j_1}\ldots \txi^{j_\k}  \tp_{i_1}\ldots \tp_{i_k}, 
\end{eqnarray} 
where $P^{i_1\dots i_k}_{j_1\ldots j_\k}\in\Ga(|\Lambda|^{\d-\k/n})$.
%$\vol_\eta$ is the canonical volume form $dx^1\wedge\ldots\wedge dx^n$ attached to the metric $\eta$. 
Remark that $|\vol_\eta|^{1/n}\txi^i=|\vol_{F\eta}|^{1/n}\xi^i$.
%where $\vol_{\mg}$ is the volume form associated to the metric $\mg=F\eta$.
In addition, we introduce another linear isomorphism, called the normal ordering, 
{\begin{eqnarray}\label{OrdreNormal}
\cN:\sS^\d & \longrightarrow &\Dslm   \\ \nonumber
P^{i_1\ldots i_k}_{j_1\ldots j_\k}(x)\txi^{j_1}\ldots\txi^{j_\k}\,\tp_{i_1}\ldots\tp_{i_k} &\longmapsto&  P^{i_1\ldots i_k}_{j_1\ldots j_\k}(x) \frac{\tg^{j_1}}{\sqrt{2}}\ldots\frac{\tg^{j_\k}}{\sqrt{2}}\,\frac{\hbar}{\bi}\partial_{i_1}\ldots\frac{\hbar}{\bi}\partial_{i_k},
\end{eqnarray}}
where $P^{i_1\dots i_k}_{j_1\ldots j_\k}\in\Ga(|\Lambda|^{\d})$, $\tg^j$ denotes $\g(\txi^j)$ and $\d=\m-\l$. 

\begin{prop}\label{prop:bbL-sL-ccL}
The maps $\cF$ and $\cN$ are isomorphisms of $\ce(p,q)$-modules and 
satisfy $\ve_{k,\k}\circ\cF=\Id$ on $\sT^\d_{k,\k}$ and 
$\sigma_\ell\circ\cN=\left(\frac{\hbar}{\bi}\right)^{\ell/2}\Id$ on $\sS^\d_{[\ell]}$. 
Moreover, for any $X\in\fkg$, we have 
\begin{eqnarray}\label{EqLtLd}
\cF^{-1}\Ld\cF&=& \Lt-\frac{\hbar}{2\bi} \xi_k\xi^j(\partial_i\partial_j X^k)\partial_{p_i},\\ \label{EqLdLD}
\qquad(\cN\cF)^{-1}\LD\cN\cF&=&\cF^{-1}\Ld\cF+\frac{\hbar}{4\bi}(\partial_j\partial_kX^i)\left(-2p_i\partial_{p_j}+\chi^j_i\right)\partial_{p_k}
-\frac{\hbar}{\bi}\l\partial_j(\partial_iX^i)\partial_{p_j},
\end{eqnarray}
where $\chi^j_i=\xi^j\partial_{\xi^i}-\xi_i\partial_{\xi_j}+\half\partial_{\xi_j}\partial_{\xi^i}$.
\end{prop}

For all $X\in\fkg$, the explicit expressions of the infinitesimal actions $\Ld$ and $\LD$ 
can be deduced from the Lie derivatives $\bbL^\d_X$. %They are given below, 
In the trivialization of $\bbT^\d$ given by suitable power of $\vol_\eta$,
for $X$ a generator of $\fkg$ as in \eqref{AlgLieConf}, we have
\begin{eqnarray}\nonumber
\bbL_{X_i}^\d&=&\partial_i,\\ \label{ActionT}
\bbL_{X_{ij}}^\d&=& x_i \partial_j- x_j \partial_i-(p_i\partial_{p^j}- p_j\partial_{p^i})+\xi_i\partial_{\xi^j}-\xi_j\partial_{\xi^i},\\  \nonumber
\bbL_{X_0}^\d&=& x^i \partial_i-p_i\partial_{p_i}+\d n,\\ \nonumber
%\label{LtXi}   
\bbL^\d_{\bar{X}_i}&=&(x_jx^j\partial_i - 2x_ix^j\partial_j) +(-2p_ix_j\partial_{p_j}+2x_ip_j\partial_{p_j}+2p_k x^k\partial_{p^i})\\ \nonumber
&&
+2x_j\xi^j\partial_{\xi^i}-2\xi_i x^k\partial_{\xi^k}
-2n\d x_i.
\end{eqnarray}

%%%%%%%%%%%%%%%%%%%%%%%%%%%%%%%%%%%%%%%%%%%%%%%%%
\subsection{Definitions of the $\fkg$-equivariant superization and quantization}\label{ParDefQS}
%%%%%%%%%%%%%%%%%%%%%%%%%%%%%%%%%%%%%%%%%%%%%%%%%

The supercotangent bundle $(\cM,\om)$ of $(M,\mg)$ proves to be the phase space for a classical spinning particle on $(M,\mg)$ \cite{BMa77,Rav80,Mic10a}. The Hamiltonian action of $\fkg$ turns its space of functions $\cO(\cM)$ into the $\fkg$-module $\sS^\d$, with $\d=0$. It can be interpreted as the space of classical observables for a spinning particle, whereas the space of quantum observables of such a particle is known to be the space of differential operators acting on spinors, or more precisely the $\fkg$-module $\Dslm$ for $\l=\mu=\half$. This justifies the name {\it quantization} in the following definition.

\begin{defi}\label{DefS}
Let $\l,\m\in\bbR$ and $\d=\m-\l$. A conformally equivariant quantization is an isomorphism of $\fkg$-modules, $\cQlm:\sS^\d\rightarrow\Dslm$, which preserves the bifiltrations, $\cQlm(\sS^\d_{k,\k})\subset\sD^{\l,\m}_{k,\k}$ for all $k\in\bbN$, $\k\leq n$, and is a right inverse of the principal Hamiltonian symbol maps defined in \eqref{symbolDS}: $\sigma_{\ell}\circ\cQlm=\left(\frac{\hbar}{\bi}\right)^{\ell/2}\Id$ on $\sS^\d_{[\ell]}$. 
%Moreover, $\cQlm$ preserves the usual filtration by the order of derivations.
\end{defi}

Such a quantization extends the geometric quantization of the supercotangent bundle 
obtained in \cite{Mic10a}, which is also conformally equivariant but defined only on symbols 
of Hamiltonian degree at most $2$. By Proposition \ref{prop:bbL-sL-ccL}, 
in the local model $(\bbR^{p,q},F\eta)$ of a conformally flat manifold, 
a conformally equivariant quantization reads as
\begin{equation}\label{Eq:cQN}
\cQlm=\cN\circ\cF\circ\left(\Id+N^{\l,\m}\right)\circ\cF^{-1},
\end{equation}
where $N^{\l,\m}$ is a $\mathfrak{ce}(p,q)$-invariant linear operator on $\sT^\d$, 
which lowers the $p$-degree and strictly lowers the Hamiltonian degree.
Besides, $\cQ^{\l,\l}$ gives rise to a deformation quantization of $\sS^0$, 
the star-product being defined as usual by pull-back of the product on $\sD^{\l,\l}$ via $\cQ^{\l,\l}$. 
Notice that Fedosov's deformation quantization of symplectic supermanifolds has been investigated in \cite{Bor00}.

Regarding $\sT^\d$, this is a $\fkg$-module of geometric significance, defined as a twist by densities of the tensorial $\Vect(M)$-module $\Ga(\cS TM\otimes\Lambda T^*\bbR^n)$. If $\d=0$, it contains $\Pol(T^*M)$ as a submodule, which can be interpreted as the algebra of classical observables for a (non-spinning) particle on $(M,\mg)$, so that the following map deserves the name {\it superization}.  

\begin{defi}\label{DefQ}
Let $\d\in\bbR$. A conformally equivariant superization is an isomorphism of $\fkg$-modules, $\fS^\d:\sT^\d\rightarrow\sS^\d$, 
which preserves the bifiltrations, $\fS^\d(\sT^\d_{k,\k})\subset\sS^\d_{k,\k}$ for all $k\in\bbN$, $\k\leq n$, and  is a right inverse of the principal tensorial symbol maps defined in \eqref{symbolST}: $\ve_{k,\k}\circ\fS^\d=\Id$ on $\sT^\d_{k,\k}$.
\end{defi}

According to Proposition \ref{prop:bbL-sL-ccL}, in the local model $(\bbR^{p,q},F\eta)$ 
of a conformally flat manifold, a conformally equivariant superization reads as
\begin{equation}\label{Eq:fSN}
\fS^\d=\cF\circ\left(\Id+N^{\d}\right),
\end{equation}
where $N^{\d}$ is a $\mathfrak{ce}(p,q)$-invariant linear operator on $\sT^\d$, 
which preserves the Hamiltonian degree and strictly lowers the $p$-degree.

%%%%%%%%%%%%%%%%%%%%%%%%%%%%%%%%%%%%%%%%%%%%%%%%%
%%%%%%%%%%%%%%%%%%%%%%%%%%%%%%%%%%%%%%%%%%%%%%%%%
\section{Invariant Differential Operators on the Supercotangent Bundle}\label{SecInv}
%%%%%%%%%%%%%%%%%%%%%%%%%%%%%%%%%%%%%%%%%%%%%%%%%
%%%%%%%%%%%%%%%%%%%%%%%%%%%%%%%%%%%%%%%%%%%%%%%%%

In all this section we work over the local model $(\bbR^{p,q},F\eta)$ of a conformally flat manifold, 
so that we get global actions of the Lie algebras $\e(p,q)\leq\ce(p,q)\leq\fkg$ on $\sT^\d$. 
They integrate  into actions of the corresponding Lie groups 
$\rE(p,q)\leq\CE(p,q)\leq G$, where  
$\rE(p,q)$ is the group of all isometries of the metric $\eta$ and
$G:=\rO(p+1,q+1)$ acts only locally.

The algebra of differential operators $\D(\bbR^{p,q};\bbT^0,\bbT^0)$
is isomorphic to the algebra $\D(\cM)$ of scalar differential operators on $\cM$. 
We name fiberwise operators those linear operators which act only along the fibers of the tensor bundle~$\bbT^0$.

\begin{defi}
We denote by $\D^{\d,\d'}(\cM)$ the local $G$-module structure on $\D(\cM)$ such that
$\D^{\d,\d'}(\cM)\cong\D(\bbR^{p,q};\bbT^\d,\bbT^{\d'})$ as $G$-modules.
For $H$ a Lie subgroup of $G$, the subspace of (locally) $H$-invariant differential operators is denoted by $\D^{\d,\d'}(\cM)^H$.
\end{defi}

The invariance w.r.t.\ one of the groups $\rE(p,q)$, $\CE(p,q)$ or $G$,
is equivalent to invariance w.r.t.\ the corresponding Lie algebra
and a transformation
$\tau\in\rO(p,q)$ such that $\tau^2=\Id$ and $\det\tau=-1$.
For $\tau$, one may choose the linear transformation specified by
$\tau(x^1)=-x^1$ and $\tau(x^j)=x^j$, for all $j=2,\ldots,n$.
It acts similarly on the fiberwise coordinates  $(p_i)$ and $(\xi^i)$ on $\cM$, 
and on the associated derivatives $(\partial_i,\partial_{p_i},\partial_{\xi^i})$.

%%%%%%%%%%%%%%%%%%%%%%%%%%%%%%%%%%%%%%%%%%%%%%%%%
\subsection{Euclidean invariants}
%%%%%%%%%%%%%%%%%%%%%%%%%%%%%%%%%%%%%%%%%%%%%%%%%

The superspace $\bbR^{2n|p,q}$, $m=p+q$, is endowed with a canonical Poisson bivector $\Pi=\sum_{a=1}^n \partial_a\wedge\partial_{a+n}+\sum_{a=2n+1}^{2n+p} \partial_a\wedge\partial_a-\sum_{a=2n+p+1}^{2n+m} \partial_a\wedge\partial_a$ in a Cartesian coordinate system  $(x^a)_{a=1,\ldots,2n+m}$.
The Heisenberg Lie superalgebra $\fh(2n|p,q)$ 
and the orthosymplectic Lie superalgebra $\spo(2n|p,q)$ 
are defined as Lie subalgebras of the Poisson superalgebra 
of polynomial functions over $\bbR^{2n|p,q}$. 
Namely, the first one is given by constant and linear functions, and the second one by quadratic functions. 
For $m=0$, we recover the usual definitions of the Heisenberg and symplectic Lie algebras, and $\spo(0|p,q)$ identifies with 
the orthogonal Lie algebra $\fo(p,q)$.

Using Weyl's theory of invariants \cite{Wey97}, we get below a slight generalization of the Howe duality 
between the Lie (super-)algebras $\fo(p,q)$ and $\spo(2|1,1)$ in $\spo(2n|n,n)$, $n=p+q$. 
See e.g.\ \cite{How89,LHo10} for more informations on the latter dual pair.

\begin{prop}\label{InvIso}
Let $\cE=p_i\partial_{p_i}$, $\Sigma=\xi^i\partial_{\xi^i}$ and $\d\in\bbR$. The subalgebra of isometry invariant differential operators on $\cM$ satisfies $\D^{\d,\d}(\cM)^{\rE(p,q)}\cong\mathfrak{U}(\spo(2|1,1)\ltimes\mathfrak{h}(2|1,1))$ for $n\geq 3$. For all $n\in\bbN^\times$, it is generated by the fiberwise operators 
\begin{equation}\label{EqInvIsoSpo}
\begin{array}{cccc}
R=\eta^{ij}p_ip_j,  & E=\cE+\frac{n}{2}, & T=\eta_{ij}\partial_{p_i}\partial_{p_j},&\mathsf{\Sigma}=\Sigma-\frac{n}{2},\\[3pt]
 \boldsymbol{\bDel}=\xi^i p_i, & \boldsymbol{\delta}=\eta_{ij}\xi^i\partial_{p_j},  & \boldsymbol{\delta}^*=\eta^{ij}p_i\partial_{\xi^j}, & \boldsymbol{\bDel}^*=\partial_{\xi^i}\partial_{p_i},
\end{array}
\end{equation}
generating the Lie superalgebra $\spo(2|1,1)$, and by the differential operators 
\begin{equation}\label{EqInvIsoH}
\begin{array}{ccc}
G=\eta^{ij}p_i\partial_j, & D= \partial_i\partial_{p_i}, & L=\eta^{ij}\partial_i\partial_j,\\[3pt]
\boldsymbol{d}=\xi^i\partial_i, & \boldsymbol{d}^*=\eta^{ij}\partial_{\xi^i}\partial_j, &
\end{array}
\end{equation}
generating the Lie superalgebra $\fh(2|1,1)$.
\end{prop}

\begin{proof}
We use the Weyl quantization of $T^*\cM$ which is $\spo(4n|n,n)\ltimes\mathfrak{h}(4n|n,n)$-equivariant. It establishes a correspondence between the $\e(p,q)$-invariants of $\D^{\d,\d}(\cM)$ and of $\Pol(T^*\cM)$. The latter is generated 
by the coordinates $(x^i,\hat{x}_i,p_i,\hat{p}^i,\xi^i,\hat{\xi}_i)$ of $T^*\cM$, the hat denoting the conjugate coordinates.
Regarding the action of $\e(p,q)$ on $\cM$, given by \eqref{ActionT}, the $\e(p,q)$-invariants of $\Pol(T^*\cM)$ reduce to the $\fo(p,q)$-invariant polynomials on the space $(\bbR^n)^*\times(\bbR^n)^*\times\bbR^n\times\Pi\bbR^n\times(\Pi\bbR^n)^*$ with coordinates $(\hat{x}_i,p_i,\hat{p}^i,\xi^i,\hat{\xi}_i)$. By Weyl's Theorem \cite[Theorem 2.9.A, p.53]{Wey97} these invariants split into even and odd invariants: the even ones are generated by the $15$ scalar products between these five types of coordinates, and the odd ones are constructed from the determinant and are not invariant under the group~$\rO(p,q)$. So we are left with the even invariants.
 The two squares of odd variables $\xi$ and $\hat{\xi}$ vanish and only $13$ non-vanishing scalar products remain. 
Their linear span is stable under the canonical Poisson bracket on $T^*\cM$.
We give degree $0$ to the variables $\hat{x}$.
The $8$ invariants without $\hat{x}$ generate a Lie superalgebra isomorphic to the Lie superalgebra
of quadratic polynomials on $T^*\bbR^{1|1}$, i.e. $\spo(2|1,1)$.
The $5$ other invariants correspond to polynomials of degrees $0$ and $1$. Hence, they generate 
 a Lie superalgebra isomorphic to $\fh(2|1,1)$.
By Weyl quantization, these invariants are send to the $8$ operators in \eqref{EqInvIsoSpo} and the $5$ operators in \eqref{EqInvIsoH} respectively. By equivariance property of Weyl quantization we get that the operators in \eqref{EqInvIsoSpo} generate the Lie algebra $\spo(2|1,1)$, the operators in \eqref{EqInvIsoH} generate the Lie algebra $\fh(2|1,1)$, and together they generate the algebra of even $\e(p,q)$-invariants, or equivalently of $\rE(p,q)$-invariants, in $\D^{\d,\d}(\cM)$.

Assuming $n\geq 3$, we get $d(\hat{x}^i\hat{x}_i)\wedge d(\hat{x}^ip_i)\wedge\ldots\wedge d(\xi^i\hat{\xi}_i)\neq 0$. Thus the $13$ obtained operators are algebraically independent and $\D^{\d,\d}(\cM)^{\rE(p,q)}$ is then isomorphic to the enveloping algebra $\mathfrak{U}(\spo(2|1,1)\ltimes\mathfrak{h}(2|1,1))$.
\end{proof}

All of the $13$ operators in \eqref{EqInvIsoSpo} and \eqref{EqInvIsoH} identify to well-known operators on tensors. The three operators $R,\cE,T$, restricted to $\Ga(\cS T\bbR^n)$, correspond to the metric, the Euler operator and the trace. 
As for $G,D,L$, they generalize the gradient, the divergence and the Laplacian. 
These $6$ operators have been introduced in \cite{DLO99,DOv01} as building blocks of the conformally equivariant quantization of cotangent bundles. The operator $\Sigma$ is the Euler operator on $T[1]\bbR^n$ and the $6$ remaining operators square to zero and identify to well-known (co-)differentials. Indeed, $\bd,\bd^*$ are the usual de Rham (co-)differentials on $\Ga(\Lambda T^*\bbR^n)$ and $\bdel,\bdel^*$ are Koszul type (co-)differentials on the whole space $\Ga(\cS T\bbR^n\otimes\Lambda T^*\bbR^n)$. As for $\bDel,\bDel^*$, they are (co-)differentials introduced by Manin to give a cohomological description of the Berezinian \cite{Man88}. Notice that $\bDel$ is also a function on $\cM$: this is the principal symbol of the Dirac operator. The notation $\bDel$ is borrowed from physics where it plays the role of a supercharge \cite{GRV93}. Let us mention that similar algebras of invariant operators on tensor (and spinors) have been investigated in \cite{HWa05,HWa07,HWa08}, over a background $(M,\mg)$ of constant curvature.

The bold operators with a $*$-exponent lower the degree in $\xi$ by one, while the other ones rise it by one. The next proposition shows that, indeed, $\bdel^*,\bDel^*$ are the codifferentials of $\bdel,\bDel$.

\begin{prop}\label{Prop:proj}
On $\sT^\d_{>0}:=\bigoplus_{k\geq 1,\k}\sT^\d_{k,\k}$, the two pairs of operators $\left(\frac{1}{\cE+\Sigma}\bdel\bdel^*, \frac{1}{\cE+\Sigma}\bdel^*\bdel\right)$  and 
$\left(\frac{1}{n+\cE-\Sigma}\bDel^*\bDel, \frac{1}{n+\cE-\Sigma}\bDel\bDel^*\right)$ are two pairs of complementary projections. Moreover, restricting $\bdel,\bdel^*,\bDel,\bDel^*$ to $\sT^\d_{>0}$, we get
$$
\sT^\d_{>0}=\ker\bdel\oplus\ker\bdel^* \qquad \text{and}\qquad \sT^\d_{>0}=\ker\bDel\oplus\ker\bDel^*,
$$
with $\ker\bdel=\im\bdel$, $\ker\bdel^*=\im\bdel^*=\bdel^*\ker\bdel$ and $\ker\bDel=\im\bDel$, $\ker\bDel^*=\im\bDel^*=\bDel^*\ker\bDel$.
\end{prop}
\begin{proof}
We work over $\sT^\d_{>0}$, so that the operators $\cE+\Sigma$ and $n+\cE-\Sigma$ are invertible.

A straightforward computation leads to $[\bdel^*,\bdel]=\cE+\Sigma$. Hence, $\frac{1}{\cE+\Sigma}\bdel\bdel^*$ is a projection whose complementary projection is $\frac{1}{\cE+\Sigma}\bdel^*\bdel$. 
As a consequence, we have $\ker\frac{1}{\cE+\Sigma}\bdel\bdel^*=\im\frac{1}{\cE+\Sigma}\bdel^*\bdel$ and a complementary space  is provided by $\im\frac{1}{\cE+\Sigma}\bdel\bdel^*=\ker\frac{1}{\cE+\Sigma}\bdel^*\bdel$. 

We compute now the latter four spaces. Using the above commutation relation, we first get that $\bdel^*\bdel\bdel^*=(\cE+\Sigma)\bdel^*$. The equalities of their kernels and of their images imply that
$\ker \frac{1}{\cE+\Sigma}\bdel\bdel^*=\ker\bdel^*$ and $\im\frac{1}{\cE+\Sigma}\bdel^*\bdel=\im\bdel^*$. 
Similarly, we have $\bdel\bdel^*\bdel=(\cE+\Sigma)\bdel$, which leads to $\im \frac{1}{\cE+\Sigma}\bdel\bdel^*=\im\bdel$ and $\ker\frac{1}{\cE+\Sigma}\bdel^*\bdel=\ker\bdel$. Combining all these equalities, we deduce the results concerning $\bdel,\bdel^*$ as well as the announced decomposition of $\sT^\d_{>0}$. The case of $\bDel,\bDel^*$ is analogous.
\end{proof}

%%%%%%%%%%%%%%%%%%%%%%%%%%%%%%%%%%%%%%%%%%%%%%%%%
\subsection{Conformal invariants}\label{ParConfInv}
%%%%%%%%%%%%%%%%%%%%%%%%%%%%%%%%%%%%%%%%%%%%%%%%%

The $\rE(p,q)$-action on $\sT^\d$ is independent of $\d$, hence the subspace of
$\rE(p,q)$-invariants in $\D^{\d,\d'}(\cM)$ identifies with the previously determined one in $\D^{\d,\d}(\cM)$. 
In contradistinction, the action of the dilation vector field $X_0$ (see \eqref{ActionT}) 
depends on the shift of weight~$\d'-\d$. It turns out that all the $\rE(p,q)$-invariant 
generators are also $\CE(p,q)$-invariants for a well-chosen shift. Moreover, the fiberwise operators are automatically invariant under conformal inversions. The situation for the $13$ generators of $\D^{\d,\d'}(\cM)^{\rE(p,q)}$ is sum up in the following table,
\begin{equation}\label{tableInv}
\begin{array}{|c||c|c|c|c|c|}
				\hline	
\text{values of }n(\d'-\d)	& -2 & -1 & 0 & 1& 2\\[3pt]\hline
\CE(p,q)\text{-invariant operators}	& T & \bdel, \bDel^* & \cE, D, \Sigma & \bDel, \bdel^*, \bd, \bd^* & R, G, L \\[3pt]\hline 
G\text{-invariant operators}	& T &  \bdel, \bDel^* & \cE,  \Sigma & \bDel, \bdel^* & R \\[3pt]\hline 
\end{array}
\end{equation}  
Restricting now to $\d=\d'$, we deduce that any $\CE(p,q)$-invariant operator is 
a linear combination of monomials of the form
\begin{equation}\label{cQeucl}
 R^r \bDel^\a  (\bdel^*)^\b \bd^\g (\bd^*)^{\g'} G^g  L^l \bdel^{\b'} (\bDel^*)^{\a'}  T^t D^a\Sigma^b\cE^c,
\end{equation}
where $a,b,c$ are arbitrary integers and $2(r+g+l-t)+\a+\b+\gamma+\g'-\b'-\a'=0$. 
The exponents of odd operators are equal to $0$ or $1$, since they have null square. 
As for $G$-invariant differential operators on $\cM$, we have the following description.

\begin{prop}\label{PropInvConfT}
 The algebra of conformal invariants $\D^{\d,\d}(\cM)^G$ is generated by
\begin{equation}\label{ConfGen}
 \cE, \Sigma, RT, \bDel\bDel^*, \bdel^*\bdel, \bDel\bdel, \bdel^*\bDel^*.
\end{equation}
It coincides with the algebra of fiberwise $\CO(p,q)$-invariant operators in $\D^{\d,\d}(\cM)$.
Moreover, $\cE$ is in the center and $\cE,\Sigma,RT,\bDel\bDel^*+\bdel^*\bdel$ generate an abelian subalgebra.
\end{prop}

\begin{proof}
By a direct generalization of \cite[Lemma 3.3]{DEO04}, all conformal invariants in $\D^{\d,\d}(\cM)$ 
are fiberwise operators. 
Hence, they coincide with the $\CE(p,q)$-invariant fiberwise operators. 
Since translations act trivially on fiberwise operators, they coincide 
also with the $\CO(p,q)$-invariant fiberwise operators.
According to Formula \eqref{cQeucl}, they are generated by $\cE, \Sigma$ and $R^r \bDel^\a (\bdel^*)^\b  \bdel^{\b'} (\bDel^*)^{\a'}  T^t$ for $2r+\a+\b-\b'-\a'-2t=0$. Since the exponents of odd operators are equal to $0$ or $1$, we obtain the announced generators in \eqref{ConfGen}, plus $R\bdel\bDel^*$ and~$\bDel\bdel^* T$.

We trivially check that $\cE$ is in the center and that $\Sigma$ commute to $RT$, $\bDel\bDel^*$ and $\bdel^*\bdel$. A direct computation shows that $\half[RT,\bDel\bDel^*]=-\half[RT,\bdel^*\bdel]=[\bDel\bDel^*,\bdel^*\bdel]=R\bdel\bDel^*-\bDel\bdel^* T$. Therefore, $\cE,\Sigma,RT,\bDel\bDel^*+\bdel^*\bdel$ generate a commutative algebra. Moreover, the sum of the three commutators of $[RT,\bDel\bDel^*]$ with $RT$, $\bDel\bDel^*$, $\bdel^*\bdel$, leads to $R\bdel\bDel^*+\bDel\bdel^* T$, hence the seven operators given in \eqref{ConfGen} generate indeed $\D^{\d,\d}(\cM)^G$.
\end{proof}

 %%%%%%%%%%%%%%%%%%%%%%%%%%%%%%%%%%%%%%%%%%%%%%%%%%%%%%%%%%%%
\subsection{Three Casimir operators}\label{ParCasimir}
%%%%%%%%%%%%%%%%%%%%%%%%%%%%%%%%%%%%%%%%%%%%%%%%%%%%%%%%%%%%

Given a representation $\rho:\fkg\rightarrow\End(V)$, we recall that the Killing form $B$ of $\fkg$ induces a particular $\fkg$-invariant operator on $V$, called the Casimir operator. It is defined by
$
C_\rho=B^{\a\b}\rho(X_\a)\rho(X_\b),   
$
 where $B^{\a\b}$ is the inverse of the Gram matrix of the Killing form in the basis $(X_\a)$ of $\fkg$. 
Hence, for each of the $\fkg$-modules
$\sT^\d$, $\sS^\d$ and $\Dslm$, we get a Casimir operator. 
 We can pull-back the ones of $\sS^\d$ and $\Dslm$ to~$\sT^\d$, via the isomorphisms $\cF$ and $\cN$ defined in \eqref{Tenseur_Poly_Moments} and \eqref{OrdreNormal} respectively. We get then $\ce(p,q)$-invariant operators on $\sT^\d$. They can be written in terms of the generators listed in \eqref{EqInvIsoSpo} and \eqref{EqInvIsoH}. The Casimir operator of $\sT^\d$ is $\fkg$-invariant and can more specifically be written in terms of the operators in \eqref{ConfGen}. 

\begin{prop}\label{PropCalculCasimir}
Let $\d=\m-\l$. The  Casimir operators of the three $\fkg$-modules $\sT^\d$, $\sS^\d$ and $\Dslm$ 
read on $\sT^\d$ respectively as
\begin{eqnarray}\label{tC}
\tC &=&\hat{C} +\Sigma(\Sigma-n)+2(\bDel\bDel^*+\bdel^*\bdel)-2\cE,\\ \nonumber %\label{Casimird}
\Cd &=& \tC +2\frac{\hbar}{\bi}\bd\bdel,\\ \nonumber %\label{Casimirlm}
\Clm &=& \tC +\frac{\hbar}{\bi}\left(GT-2\left(\cE+n\l+\half\right)D+2\bd\bdel+\bd^*\bdel+\bd\bDel^*+\half\bd^*\bDel^*\right),
\end{eqnarray}
where $\hat{C}=RT+\left[1+n(\d-1)-\cE\right]\cE-n^2\d(\d-1)$ is the Casimir operator of the $\fkg$-module $\Ga(\cS T\bbR^n\otimes|\Lambda|^\d)$.
\end{prop}
\begin{proof}
Using the basis of $\fkg$ introduced in \eqref{AlgLieConf} and the computation of the Gram matrix of the Killing form performed in \cite{DLO99}, we deduce the general expression of the Casimir operator, 
$$
C_\rho=\half\eta^{ik}\eta^{jl}\rho(X_{ij})\rho(X_{kl})-\rho(X_0)^2-\half\eta^{ij}\rho(X_i)\rho(\bar{X}_j)-\half\eta^{ij}\rho(\bar{X}_i)\rho(X_j), 
$$
for any representation $\rho$ of $\fkg$. It suffices then to replace $\rho$ by successively the three representations $\bbL^\d$, $\sL^\d$ and $\ccL^{\l,\m}$ and to apply Formul{\ae}  \eqref{EqLtLd}-\eqref{ActionT}. The results follow after straightforward but rather lengthy computations. The formula of $\hat{C}$ is obtained in \cite{DLO99}.
\end{proof}

%%%%%%%%%%%%%%%%%%%%%%%%%%%%%%%%%%%%%%%%%%%%%%%%%
%%%%%%%%%%%%%%%%%%%%%%%%%%%%%%%%%%%%%%%%%%%%%%%%%
\section{Main results}
%%%%%%%%%%%%%%%%%%%%%%%%%%%%%%%%%%%%%%%%%%%%%%%%%
%%%%%%%%%%%%%%%%%%%%%%%%%%%%%%%%%%%%%%%%%%%%%%%%%

We work on a spin manifold $(M,\mg)$ of dimension $n$ and signature $(p,q)$, 
that is assumed to be conformally flat except in Sect.\ \ref{ParDEcompoT}. 
In the conformally flat case, we again denote by~$\fkg$ 
the (local) conformal Lie algebra and by $G$ the (local) conformal Lie group.

%%%%%%%%%%%%%%%%%%%%%%%%%%%%%%%%%%%%%%%%%%%%%%%%%
\subsection{Irreducible decomposition of the tensorial symbol bundle}\label{ParDEcompoT}
%%%%%%%%%%%%%%%%%%%%%%%%%%%%%%%%%%%%%%%%%%%%%%%%%

Since the $\rO(p,q)$-invariant operators introduced in \eqref{EqInvIsoSpo}  are fiberwise, they generalize to arbitrary pseudo-Riemannian manifold $(M,\mg)$, up to replacing the metric $\eta$ by $\mg$. Their commutation relations, given in \eqref{CommRel}, remain the same.

In the case of scalar differential operators, the symbols are sections of the vector bundle $\cS TM$. 
The decomposition $\cS TM=\bigoplus_k\cS^k TM$ coincides with the decomposition
of $\cS TM$ into eigenspaces of the Euler operator $\cE$. Each vector bundle $\cS^k TM$
can be further decomposed into eigenspaces of the fiberwise operator $RT$. The resulting vector bundles 
have $\rO(p,q)$-irreducible fibers, and
their spaces of sections carry an irreducible fiberwise $\rO(p,q)$-action.
For weighted spinor differential operators, the space of tensorial symbols is $\sT^\d=\Ga(\bbT^\d)$, cf.\ Definition \ref{defi:T}. 
An analogous decomposition of $\bbT^\d$ can by obtained by decomposing its fibers
$\cS\bbR^n\otimes\Lambda(\bbR^n)^*$ into irreducible representations of $\rO(p,q)$. 
This generalization of harmonic decomposition has been carried out in \cite{Hom01}. We recover it here independently as the decomposition into joint eigenspaces of the commuting $\rO(p,q)$-invariant operators $\cE,\Sigma,RT,\bDel\bDel^*+\bdel^*\bdel$, obtained in Proposition~\ref{PropInvConfT}. To take into account the weight~$\d$, we rather regard these bundles 
endowed with the fiberwise $\CO(p,q)$-action.
We introduce the projection $\Pi_0:\ker T^2\rightarrow\ker T$, given by $\Id-\frac{1}{4n+2\cE}RT$, 
and the operator $\bDel_0=\Pi_0\circ\bDel|_{\ker T}$.
\begin{thm}\label{DiagCt}
The vector bundle $\bbT^\d$ admits 
a decomposition into vector bundles with $\CO(p,q)$-irreducible fibers,
which reads as follows on its space of sections
\begin{equation}\label{DecompoT}
\sT^\d=\bigoplus_{k\in\bbN}\, \bigoplus_{\k\leq n}\bigoplus_{s\leq\lfloor k/2\rfloor}\bigoplus_{\a,\b\in\{0,1\}}\sT^\d_{k,\k,s;\a\b},
\end{equation}
where
\begin{equation}\label{Tkksab}
\sT^\d_{k,\k,s;\a\b}=\sT^\d_{k,\k}\bigcap \left(R^s(\bDel_0)^\a(\bdel^*)^\b\cdot\left(\ker T\cap\ker \bdel\cap\ker\bDel^*\right)\right),
\end{equation}
The multiplicity in the Decomposition \eqref{DecompoT} is at most two, the only 
linear isomorphisms intertwining the fiberwise $\CO(p,q)$-action being
 \begin{equation}\label{SplitOmDelta}
\xymatrix{
 \sT^\d_{k,\k,s;10} \ar[rr]^{\bdel^*\bDel^*} & &  \sT^\d_{k,\k-2,s;01} \ar@/^1pc/[ll]^{\bDel\bdel}, &\text{and} & 
 \sT^\d_{k,\k,s;00} \ar[rr]^{\bDel_0\bdel^* T}& & \sT^\d_{k,\k,s-1;11} \ar@/^1pc/[ll]^{R\bdel\bDel^*}, 
}
\end{equation}
for all $k,\k,s$ such that source and target spaces are well-defined.
If  $(M,\mg)$ is conformally flat, $\sT^\d_{k,\k,s;\a\b}$ is a $G$-module. 
This is an eigenspace of the Casimir operator $\tC$ with eigenvalue
\begin{equation}\label{gkksab}
\g_{k,\k,s;\,\a\b}=\hat{\g}_{k,s}+\k(\k-n)+2(\a+\b-1)(k-2s)+2(\b-\a)\k+2\a(n-2\b),
\end{equation}
where $\hat{\g}_{k,s}=2s[n+2(k-s-1)]+2k[1+n(\d-1)-k]-n^2\d(\d-1)$ is the eigenvalue of the Casimir operator $\hat{C}$, see \eqref{tC}.
\end{thm}
\begin{proof}
We first obtain the Decomposition \eqref{DecompoT} by decomposing $\sT^\d$ into joint eigenspaces of the four commuting operators $\cE,\Sigma,RT,\bDel\bDel^*+\bdel^*\bdel$. 
The decomposition of $\sT^\d$ into joint eigenspaces of  $\cE,\Sigma,RT$ is 
a straightforward generalization of the usual harmonic decomposition, 
$$
\sT^\d=\bigoplus_{k,\k,s}\sT^\d_{k,\k}\bigcap\left(R^s\ker T \right).
$$
It remains to decompose the space $\ker T$ into eigenspaces of $\bDel\bDel^*+\bdel^*\bdel$. 
This is tricky since $RT$, $\bdel^*\bdel$ and $\bDel\bDel^*$ do not commute.
As $\sT^\d_{0,\k}\subset \ker T\cap\ker \bdel\cap \ker\bDel^*$, 
we restrict below to $\sT^\d_{>0}=\bigoplus_{k\geq 1,\k}\sT^\d_{k,\k}$. 
%We have $\sT^\d=\sT^\d_{0,0}\oplus\sT^\d_{>0}$ , with $\sT^\d_{>0}=\bigoplus_{k+\k\geq 1}\sT^\d_{k,\k}$
Thus, the operator $\cE+\Sigma$ is invertible and Proposition \ref{Prop:proj} applies.
As a consequence, $\frac{1}{\cE+\Sigma}\bdel\bdel^*$ is a projection on  $\ker \bdel$ along $\bdel^*\ker \bdel$. 
Since $[T,\bdel\bdel^*]=2\bdel\bDel^*$ and $[\bDel^*,\bdel\bdel^*]=-2\bDel^*+\bdel^*T$, the latter projection preserves the space $\ker\bDel^*\cap\ker T$ and induces the splitting
$$
\ker\bDel^*\cap\ker T=\left(\ker T\cap\ker \bdel\cap\ker\bDel^*\right)\oplus \bdel^*\left(\ker T\cap\ker \bdel\cap\ker\bDel^*\right).
$$
The Decomposition \eqref{DecompoT} follows then from the equality 
$\ker T=\bigoplus_{\a=0,1}(\bDel_0)^\a\ker T\cap\ker\bDel^*$. 
The proof of this equality boils down to the proof of existence of an operator $A$ such that
 \begin{equation}\label{Seq:Split}
\xymatrix{
 0\ar[r]& \ker \bDel^*\cap\ker T \ar[r] &\ker T \ar[rr]^{A\bDel^*\quad}& & \ker \bDel^*\cap\ker T \ar@/^1pc/[ll]^{\bDel_0} \ar[r]& 0
}
\end{equation}
is a split exact sequence.
By Proposition \ref{Prop:proj} we get $\ker\bDel^*\cap\ker T=\bDel^*\ker T$, 
and Table \eqref{CommutationkerT} leads to the equality $\bDel^*\bDel_0 = (n+\cE-\Sigma)-\frac{1}{\left(n+2(\cE-1)\right)}\bdel^*_0\bdel$ on $\bDel^*\ker T$. In consequence, there exists an operator $A$ 
satisfying $A\bDel^*\bDel_0=\Id$ on $\bDel^*\ker T$. This proves \eqref{Seq:Split} and \eqref{DecompoT} follows.

Since the group $\rO(p,q)$ is semi-simple, the space $\sT^\d_{k,\k,s;\a\b}$
splits into irreducible pieces under the fiberwise $\rO(p,q)$-action.  
Restricted to such a space, the eight generators of the fiberwise $\rO(p,q)$-invariant operators, 
defined in \eqref{EqInvIsoSpo}, are either null or have zero kernel. This means 
that $\sT^\d_{k,\k,s;\a\b}$ is irreducible for the fiberwise 
$\rO(p,q)$-action. Clearly, it carries also a fiberwise $\CO(p,q)$-action.
As the modules $\sT^\d_{k,\k,s;\a\b}$ are joint eigenspaces for the four operators $\cE,\Sigma,RT,\bDel\bDel^*+\bdel^*\bdel$, the only isomorphisms between them come from the remaining fiberwise $\CO(p,q)$-invariant operators listed in Proposition \ref{PropInvConfT}. The Isomorphisms \eqref{SplitOmDelta} follow. 

If  $(M,\mg)$ is conformally flat, the operators entering in the definition of $\sT^\d_{k,\k,s;\a\b}$
are $G$-invariant, hence this is a $G$-module.
Direct computations lead to the eigenvalue of the Casimir operator $\tC$.
\end{proof}

Notice that the modules of zero degree in the odd variables read as $\sT^\d_{k,0,s;01}$ if $k>0$. 
On those modules, the eigenvalues of the Casimir operators $\tC$ and $\hat{C}$ are equal.

%%%%%%%%%%%%%%%%%%%%%%%%%%%%%%%%%%%%%%%%%%%%%%%%%
\subsection{Classification of conformally invariant operators on $\sT^\d$}\label{ParConfInvLoc}
%%%%%%%%%%%%%%%%%%%%%%%%%%%%%%%%%%%%%%%%%%%%%%%%%

In Sect.\ \ref{ParConfInv}, we have determined all the 
$G$-invariant differential operators acting on $\sT^\d$, over the 
local model $(\bbR^{p,q},F\eta)$. 
We turn now to those which are linear but not necessarily differential. 

Let $k,l\in\bbN$. According to \cite{LOv99}, a linear operator $\Ga(\cS^k T \bbR^n)\to\Ga(\cS^l T\bbR^n)$,
which is invariant under translations and dilation, is a local operator.
Hence, this is the restriction of a differential operator on $T^*M$.
This result extends straightforwardly to our context. 
Adapting the proof of Proposition \ref{InvIso}, we obtain
\begin{prop}\label{ce-InvDiff}
Any $\CE(p,q)$-invariant linear operator $A:\sT^\d_{k,\k,s;\a\b}\rightarrow\sT^{\d'}_{k',\k',s';\a'\b'}$ 
coincides with the restriction of a $\CE(p,q)$-invariant differential operator in $\D^{\d,\d'}(\cM)$.
\end{prop} 
As a consequence, $G$-invariant linear operators like $A$ are restrictions of
$\CE(p,q)$-invariant differential operators in $\D^{\d,\d'}(\cM)$. However, they do not 
 always coincide with restriction of $G$-invariant differential operators in $\D^{\d,\d'}(\cM)$.
To classify them, we restrict ourselves to fiberwise irreducible  bundles.
According to the previous section, we have the following commutative diagram of $G$-modules

\begin{equation}\label{Diag:00}
\xymatrix{
\sT^{\d}_{k,\k,s;\a\b}\ar[d]_{\bdel^\a(\bDel^*)^\b T^s}\ar[rr] && \sT^{\d'}_{k',\k',s';\a'\b'}\\
\sT^{\d_0}_{k_0,\k_0,0;00}\ar[rr] && \sT^{\d'_0}_{k'_0,\k'_0,0;00}\ar[u]_{R^{s'}(\bDel_0)^{\a'}(\bdel^*)^{\b'}}
}
\end{equation}
for well-chosen degrees and density weights on the bottom part. The vertical arrows being isomorphisms, the general classification of conformally invariant operators acting on $\sT^\d$ boils down to the one on its irreducible pieces of the form $\sT^\d_{k,\k,0;00}$. They are given by certain linear combinations of $\CE(p,q)$-invariant operators of the type
\begin{equation}\label{InvMonomial}
\bd_{\bf 0}^{\g'}(\bd^*_{\bf 0})^{\g} G_{\bf 0}^g L^\ell_{\bf 0} D^d_{\bf 0} ,
\end{equation}
where the index ${\bf 0}$ denotes the restriction and corestriction to $\ker T\cap\ker \bDel^*\cap\ker\bdel$.  
We determine all such $G$-invariant operators below. Using their equivalent description as morphisms of generalized Verma modules, their classification can also be derived from the general statements in \cite{BCo85a,BCo85b}.

\begin{thm}\label{thmConfInv}
Let $k,k'\in\bbN$, $0\leq\k,\k'\leq n$, $\d,\d'\in\bbR$. Set $j=n(\d'-\d)$. 
Over $(\bbR^{p,q},F\eta)$, the space of conformally invariant 
linear operators $\Hom(\bbT^\d_{k,\k,0;00},\bbT^{\d'}_{k',\k',0;00})^G$ 
  is either trivial or of dimension~$1$, generated by
\begin{itemize}
\item $D_{\bf 0}^d$ if $k'-k=-d$, $\k'=\k$, $j=0$ and $\d=1+\frac{2k-d}{n}$,
\item $G_{\bf 0}^g$ if $k'-k=g$, $\k'=\k$, $j=2g$ and $\d=-\frac{g}{n}$,
\item $\ccL_\ell$ if $k'=k$, $\k'=\k$, $j=2\ell$ and $\d=\half+\frac{k-\ell}{n}$,
\item $\bd_{\bf 0}$ if $k'=k$, $\k'-\k=1$, $j=1$ and $\d=\frac{k+\k}{n}$,
\item $\bd^*_{\bf 0}$ if $k'=k$, $\k'-\k=-1$, $j=1$ and $\d=1+\frac{k-\k}{n}$,
\end{itemize}
where $\ccL_\ell=\sum_{\ve=0,1}a_{\ve,j}\bd_{\bf 0}^\ve(\bd^*_{\bf 0})^\ve \sum_{j=0}^\ell  G_{\bf 0}^jL_{\bf 0}^{\ell-\ve-j}D_{\bf 0}^j$ for some coefficients $a_{\ve,j}\in\bbR$. 
If $n$ is even, we get also
\begin{itemize}
\item $\bd_{\bf 0}\ccL_\ell$ if $\k=n/2+\ell$, $k'=k$, $\k'-\k=1$, $j=2\ell+1$ and $\d=\half+\frac{k-\ell}{n}$,
\item $\bd^*_{\bf 0}\ccL_\ell$ if $\k=n/2-\ell$, $k'=k$, $\k'-\k=-1$, $j=2\ell+1$ and $\d=\half+\frac{k-\ell}{n}$.
\item $\bd_{\bf 0}\bd^*_{\bf 0}$ if $\k=n/2+1$, $k'=k$, $\k'=\k$, $j=2$ and $\d=\half+\frac{k-1}{n}$,
\item $\bd^*_{\bf 0}\bd_{\bf 0}$ if $\k=n/2-1$, $k'=k$, $\k'=\k$, $j=2$ and $\d=\half+\frac{k-1}{n}$.
\end{itemize}
\end{thm}

\begin{proof}
We follow the proof given by the author in \cite{Mic11a}, in the case of the cotangent bundle. 
Let $A\in\Hom(\bbT^\d_{k,\k,0;00},\bbT^{\d'}_{k',\k',0;00})^G$. 
According to \eqref{InvMonomial}, it is of the form $A=\bd_{\bf 0}^{\g'}(\bd_{\bf 0}^*)^{\g} G_{\bf 0}^g\ccL_\ell D_{\bf 0}^d$, where $g,d,\ell$ are integers, $\g,\g'=0,1$ and
$\ccL_\ell=\sum_{\ve=0,1}a_{\ve,j}\bd_{\bf 0}^\ve(\bd^*_{\bf 0})^\ve \sum_{j=0}^\ell  G_{\bf 0}^jL_{\bf 0}^{\ell-\ve-j}D_{\bf 0}^j$ with $a_{\ve,j}\in\bbR$. 
Moreover, we can ask for $a_{0,0}\neq 0$ so that $\ell$ is minimal. 

The operator $A$ being $\CE(p,q)$-invariant, its $G$-invariance is equivalent to
its invariance under a conformal inversion $\bar{X}_i$.
The action of $\bar{X}_i$ on the powers of the five operators entering into $A$ is computed in the appendix, it follows that
\begin{equation}\label{Egdl}
[A,\bbL_{\bar{X}_i}^*]\in E_D\partial_{p_i}\oplus \Pi_{\bf 0} p_iE_G\oplus\partial_i E_L\oplus \Pi_{\bf 0}\xi_i E_{\bd} \oplus E_{\bd^*}\Pi_{\bf 0}\partial_{\xi^i},
\end{equation}
where $\Pi_{\bf 0}$ is the projection onto the space $\ker T\cap\ker \bDel^*\cap\ker\bdel$, and $E_D$, $E_G$, $E_L$, $E_{\bd}$, $E_{\bd^*}$ are vector spaces generated by the five operators $\bd_0$, $\bd^*_0$, $G_0$,  $L_0$, $D_0$, which satisfy
$E_D D, G E_G, L E_L, \bd E_{\bd}, E_{\bd^*}\bd^* \in\D(\bbT^\d_{k,\k,0;00},\bbT^{\d'}_{k',\k',0;00})$.  
The independence of the monomials in $A$ with different powers of $L_0$ together with the vanishing of the five components of $[A,\bbL_{\bar{X}_i}^*]$ lead then to the result.
E.g. the vanishing of the component in $E_D\partial_{p_i}$ of the higher degree term in $L_{\bf 0}$ of~$[A,L_{X_i}^*]$ reads as
\[
[D_{\bf 0}^{d},L_{X_i}^*]=0.
\]
By the Relations \eqref{L0Xi}, if $d\neq 0$, the above equation is equivalent to $\d=1+\frac{2k-d-1}{n}$.
Along the same reasoning we get that $G_{\bf 0}^g$, $\bd_{\bf 0}^{\g'}$ and $(\bd^*_{\bf 0})^\g$ 
are $G$-invariant and, by \eqref{L0Xi}, we have
$\d+\frac{2l}{n}=\frac{1-g}{n}$ if $g\neq 0$, $\d+\frac{2(\ell+g)}{n}=1+\frac{k-\k}{n}$ if $\g=1$ and $\d+\frac{2(\ell+g)+\g}{n}=\frac{k+\k}{n}$ if $\g'=1$. 
As the operator~$A$ is $G$-invariant, the operator $\ccL_\ell$ is also $G$-invariant. 
Since the component in $\partial_iE_L$ of the higher degree term in~$L_0$ of $[\ccL_\ell,L_{X_i}^*]$  vanishes, Eqs.\ \eqref{L0Xi} lead us to
$\d=\half+\frac{k-\ell}{n}$ if $\ell\neq 0$. Then, straightforward but lengthy computations show that there exist 
unique reals $a_{\ve,j}$ such that $\ccL_\ell$ is conformally invariant. 

If $n$ is odd, the five values found for $\d$ are incompatible two by two, so only one of the five exponents can be non zero. If $n$ is even, the values of $\d$ for two non-vanishing exponents among $\g$, $\g'$, $\ell$ can be compatible for constrained value of $\k$. The result follows.
\end{proof}

\begin{rmk}\label{rmk:ginv}
Let $\tau\in\rO(p,q)$ such that $\tau^2=\Id$ and $\det\tau=-1$.
Under the action of $\tau$, the $\fkg$-invariant operators are preserved 
up to a global sign. Those which are fixed by $\tau$ are $G$-invariant and classified above. 
Those which are anti-fixed by $\tau$ are built from the canonical volume
form on $\bbR^{p,q}$.
They are easily proved to be 
fiberwise operators preserving the $p$-degree.
\end{rmk}

%%%%%%%%%%%%%%%%%%%%%%%%%%%%%%%%%%%%%%%%%%%%%%%%%
\subsection{Existence and uniqueness of the $\fkg$-equivariant quantization and superization}
%%%%%%%%%%%%%%%%%%%%%%%%%%%%%%%%%%%%%%%%%%%%%%%%%

In the seminal paper \cite{DLO99}, existence and uniqueness of the conformally equivariant quantization of cotangent bundles was proven, using diagonalization of the Casimir operators of the modules of differential operators and of their symbols. Thanks to Theorem \ref{DiagCt}, we can
apply the same method to prove existence and uniqueness of equivariant superization and quantization of supercotangent bundles, introduced in Definitions \ref{DefS} and \ref{DefQ}. This is our main result.

\begin{thm}\label{superequi}
Let $(M,\mg)$ be a conformally flat manifold of even dimension and $\d=\m-\l\in\bbR$. 
There exist two subsets $I^\fS,I^\cQ\subset\bbQ$ such that:
\begin{enumerate}
\item if $\d\not\in I^{\fS}$, there exists a unique conformally equivariant superization 
$
\fS^\d:\sT^\d\rightarrow \sS^\d,
$
\item if $\d\not\in I^{\cQ}$, there exists a unique conformally equivariant quantization 
$
\cQlm:\sS^\d\rightarrow \Dslm.
$
\end{enumerate}
\end{thm}

\begin{proof}
The two results can be proved in the same way, following \cite{DLO99}. We focus on the second one, as we will provide a stronger statement for superization in Theorem \ref{ProprS}. 

Existence and uniqueness of $\cQlm$ is a result of local nature so that we can work over the local conformal model $(\bbR^{p,q},F\eta)$, with $F$ a positive function. 
We denote by $\widetilde{\Cd}=\cF\circ\Cd\circ\cF^{-1}$ the Casimir operator of $\sS^\d$ and by $\widetilde{\Clm}=\cN\circ\cF\circ\Clm\circ\cF^{-1}\circ\cN^{-1}$ the one of $\Dslm$.
 The operators $\Cd$ and $\Clm$ are computed in Proposition \ref{PropCalculCasimir}. If a conformally equivariant quantization $\cQlm$ exists then it intertwines the two former Casimir operators, i.e.\ $\cQlm\circ\widetilde{\Cd}=\widetilde{\Clm}\circ\cQlm$. As a consequence, each eigenvector of $\widetilde{\Cd}$ is mapped by $\cQlm$ to an eigenvector of $\widetilde{\Clm}$ of same eigenvalue and same principal symbol.

\begin{lem}\label{lem:MethodCasimir}
Let $k,\k,s,\a,\b\in\bbN$ such that the
space $\sT^\d_{k,\k,s;\a\b}$ is well-defined.
Assume $N^\d$ is a linear operator on $\sT^\d$ which lowers by one the $p$-degree.
There exists a finite subset $I\subset\bbQ$ such that, for all $\d\in\bbR\setminus I$
and $P\in\sT^\d_{k,\k,s;\a\b}$,
there exists a unique eigenvector of $\tC+N^\d$ 
of the form $P_k+P_{k-1}+\cdots + P_0$ with $P_k=P$ and $P_l\in\bigoplus_{\k}\sT^\d_{l,\k}$.
\end{lem}

\begin{proof}
We use notation of the Lemma's statement.
According to Theorem \ref{DiagCt}, we have 
$\tC P_k= \g_{k,\k,s;\a\b} P_k$. Hence, the equality 
$(\tC+N)(P_k+\cdots+P_0)=\g (P_k+\cdots+P_0)$ implies that $\g=\g_{k,\k,s;\a\b}$ and $NP_l=(\g_{k,\k,s;\a\b}-\tC)P_{l-1}$
for all $l=k,k-1,\ldots,1$. Then, $P_k+\cdots +P_0$ is uniquely determined by $P_k$, as soon as $\g_{k,\k,s;\a\b}$ is distinct of the eigenvalues of $\tC$ on $\bigoplus_{0\leq l\leq k-1}\bigoplus_\k\sT^\d_{l,\k}$. In view of Formula \eqref{gkksab}, this is true except for a finite number of rational values of $\d$.
\end{proof}

The lemma can be applied to the Casimir operators $\Cd$ and $\Clm$. Starting with $P\in\sT^\d_{k,\k,s;\a\b}$, we get unique eigenvectors $P^\sS$ and $P^\sD$ of $\Cd$ and $\Clm$. In consequence, a conformally equivariant quantization should satisfy $\cQlm:\cF(P^{\sS})\mapsto\cN\circ\cF(P^\sD)$. 
By Decomposition \eqref{DecompoT}, this specifies a unique map on the full symbol space $\sT^\d$, 
which is indeed a $\fkg$-equivariant quantization.
By Lemma \ref{lem:MethodCasimir}, this reasoning applies for all $\d\in\bbR\setminus I$ for some $I\subset \bbQ$.
\end{proof}

The significance of the subsets of exceptional values of the weight $\d$ in the context of equivariant quantization has been revealed in \cite{Mic11a}, following previous results in \cite{Sil09}. They correspond to the existence of $\fkg$-invariant operators on the initial module, which are not fiberwise. The proof in \cite{Mic11a} generalizes straightforwardly to this context.
We use the convention $\sT^\d_{k,\k}=\{0\}=\sS^\d_{k,\k}$ if $k<0$ or $\k<0$ or $\k>n$.

\begin{thm}\label{ExistenceSuper}\cite{Mic11a}
The $\fkg$-equivariant superization exists and is unique on $\sT^\d_{k,\k,s;\a\b}$ if and only if there is no non-trivial $\fkg$-invariant operator from $\sT^\d_{k,\k,s;\a\b}$ to $\sT^\d_{k-\ell,\k+2\ell}$, for $\ell\in\bbN^*$.  \\
The $\fkg$-equivariant quantization exists and is unique on a $\fkg$-submodule of $\sS^\d_{k,\k}$ if and only if there is no non-trivial $\fkg$-invariant operator from $\sS^\d_{k,\k}$ to $\sS^\d_{k-\ell,\k+2\ell-\ell'}$, for $\ell,\ell'\in\bbN$ and $\ell+\ell'\geq 1$.
\end{thm}

From Theorem \ref{thmConfInv} and Remark \ref{rmk:ginv}, we deduce then the following corollary.

\begin{cor}\label{cor:existQ}
The subsets of critical values $I^{\fS}$ and $I^{\cQ}$ for the conformally equivariant superization and quantization are included into $\frac{1}{n}\bbN^*$.
\end{cor}

We also have the following proposition. 

\begin{prop}\label{GgEquiv}
If it exists, the $\fkg$-equivariant superization (resp.\ quantization) is in fact $G$-equivariant.
\end{prop}
\begin{proof}
Let $\tau\in\rO(p,q)$ such that $\tau^2=\Id$ and $\det\tau=-1$.
The $\fkg$-equivariant superization splits into two parts,
one is fixed by the action of $\tau$ and the other one is antifixed.
The fixed part is $G$-equivariant,
while the anti-fixed part corresponds, via the map $\cF$, 
to a $\fkg$-invariant operator on $\sT^\d$
which is anti-fixed and strictly lowers the $p$-degree.
According to Remark \ref{rmk:ginv}, this is the zero operator. 
The proof is analogous for the $\fkg$-equivariant quantization.
\end{proof}

We have defined a conjugation in $\cO(\cM)$ (see \eqref{conjugation})
and an adjoint operation on $\sD^{\l,1-\l}$ (see \eqref{adjunction}). 
The conformally equivariant quantization intertwines both.

\begin{prop}
If $\m-\l\notin I^\cQ$ and $\l+\m=1$, the conformally equivariant  quantization $\cQlm$ 
satisfies $\cQlm(\overline{P})=\cQlm(P)^*$, for all $P\in\sS^{\m-\l}$.
\end{prop}
\begin{proof}
Let $\l+\m=1$. As the linear map $\g$ intertwines conjugation and the adjoint operation, 
the map $P\mapsto\cQlm(\bar{P})^*$ is a right inverse to the principal symbol map $\sigma$ on homogeneous symbols. 
Moreover, we easily check that $\overline{\sL_X^\d P}=\sL_X^\d\overline{P}$
and $(L_X^\l)^*=-L_X^\m$, for all $X\in\fkg$ and $P\in\sS^\d$. 
Consequently, the map defined above is a $\fkg$-equivariant quantization 
and by uniqueness it is equal to $\cQlm$. We deduce that $\cQlm$ intertwines conjugation and adjoint operation.
\end{proof}

Let $\cE=p_i\partial_{p_i}$ and $\Sigma=\xi^i\partial_{\xi^i}$ be the Euler operators 
of the vector bundles $T^*M$ and $T[1]M$. Assume $\m-\l=\d\notin I^\fS\cup I^\cQ$. 
According to Proposition \ref{IsoGrDST} and Theorem \ref{superequi}, the map $\cQlm\circ\fS^\d\circ((\hbar/\bi)^{-\cE}\otimes(\sqrt{2})^\Sigma\g^{-1})$  
is a $\fkg$-module morphism 
$$
\Ga(\cS TM\otimes|\Lambda|^\d)\otimes_\Cinfty\Ga(\End S)\to\D(M;S\otimes|\Lambda|^\l,S\otimes|\Lambda|^\m).
$$
One can show this is a right inverse of the usual principal symbol maps. 
Hence, it is a $\fkg$-equvariant quantization in the sense of \cite{CSi09}.
%This gives a particular instance of the general result in \cite{CSi09} about AHS-equivariant quantization. 

%%%%%%%%%%%%%%%%%%%%%%%%%%%%%%%%%%%%%%%%%%%%%%%%%
%%%%%%%%%%%%%%%%%%%%%%%%%%%%%%%%%%%%%%%%%%%%%%%%%
\section{Explicit formul{\ae} for the $\fkg$-equivariant quantization and superization}
%%%%%%%%%%%%%%%%%%%%%%%%%%%%%%%%%%%%%%%%%%%%%%%%%
%%%%%%%%%%%%%%%%%%%%%%%%%%%%%%%%%%%%%%%%%%%%%%%%%

%%%%%%%%%%%%%%%%%%%%%%%%%%%%%%%%%%%%%%%%%%%%%%%%%
\subsection{Local formul{\ae}}
%%%%%%%%%%%%%%%%%%%%%%%%%%%%%%%%%%%%%%%%%%%%%%%%%

We restrict to the local model $(\bbR^{p,q},F\eta)$, with $p+q$ even, so that the conformally equivariant
quantization and superization read as in Eqs.\ \eqref{Eq:cQN} and \eqref{Eq:fSN}.
We provide explicit formul{\ae} for both maps in terms of the Euclidean invariants operators
(for the metric~$\eta$) classified in Sect.\ \ref{SecInv}. 
We start with the superization which is easier to deal with.
As previously, the $0$ index denotes restriction and corestriction of operators to the space $\ker T$ and $\Pi_0$ is the projection onto that space. We set $I^\fS:=\frac{1}{n}(\bbN\setminus\{0,1\})=\{\frac{2}{n},\frac{3}{n},\ldots\}$.

\begin{thm}\label{ProprS}
%Let $(M,g)$ be a conformally flat manifold of dimension $n$. 
The conformally equivariant superization $\fS^\d$ exists on $\sT^\d_{k,\k,s;\a\b}$ if and only if $\d\in \bbR\setminus I^{\fS}_{k,\k,s;\a\b}$, with $I^{\fS}_{k,\k,s;\a\b}\subset I^\fS$ a finite subset given below. If the map $\fS^\d$ exists, then it is unique. In terms of the map $\cF$ (see \eqref{Tenseur_Poly_Moments}) and the operators in \eqref{EqInvIsoSpo}--\eqref{EqInvIsoH}, it reads as 
\begin{eqnarray*}
\fS^\d &=&\cF\circ\bigg(\Id  +\frac{\hbar}{\bi}c_{\bd}\\ %\nonumber
&\times&\Big[\bd\bdel+c_D \,R^s\bDel\bdel(D+\frac{1}{k+\k-n\d}\bd\bDel^*)T^s+
c_G \,R^{s-1}\bDel\bdel(G_0+\frac{1}{k+\k-n\d}\bd\bdel^*_0)T^s \Big]\bigg),
\end{eqnarray*}
where the real coefficients are given by
\begin{equation}\label{coef_generiqueS} 
\left\{\begin{array}{l}
c_{\bd}=-\frac{1}{k+\k+1 -n\d},\\ [6pt]
c_D=-\frac{1}{(2(k-s-1)+n(1-\d))\rho_{k,s}a_{k,s}},\\ [6pt]
c_G =-\frac{2s}{(2s-n\d)\rho_{k,s}},
\end{array}\right.
\end{equation} 
with $\rho_{k,s}=\prod_{s'=1}^s2s'(n+2(k-2s+s'-1))$ the eigenvalue of $R^sT^s$ on $\sT^\d_{k,\k,s}$
($\rho_{k,0}=1$ by convention) and $a_{k,s}=\frac{n+2(k-s)}{n+2(k-s+1)}$.

The non-existence cases $\d\in I^{\fS}_{k,\k,s;\a\b}$ correspond to the existence of $\fkg$-invariant operators with source space $\sT^\d_{k,\k,s;\a\b}$ which lower the $p$-degree. They are classified below
\begin{equation}\label{tableInvSuper}
\setlength{\extrarowheight}{6pt}
\begin{array}{|c||c|c|c|c|}
\hline
\d & \fkg-\text{invariant operator}  & \text{source space $\sT^\d_{k,\k,s;\a\b}$} \\[4pt]

\hline
\frac{k+\k+1}{n} & [\bd\bdel+c_D\ldots+c_G\ldots]&  \a\b\neq 10, \k\leq (n-2), 1\leq k   \\[4pt]

\hline
\frac{2s}{n} & R^{s-1}\bDel\bdel(G_0+\frac{1}{k+\k-n\d}\bd\bdel^*_0)T^s & \a=0, 1\leq s, \k\leq (n-2)  \\[4pt]

\hline
1+\frac{2(k-s-1)}{n} & R^s\bDel\bdel(D+\frac{1}{k+\k-n\d}\bd\bDel^*)T^s &  \b=1, \k\leq (n-2), 2\leq k   \\[4pt]

\hline
\frac{k+\k}{n}& R^s\bDel\bdel\bd\bDel^* T^s & \a\b=11, 1\leq\k\leq (n-1), 2\leq k  \\[4pt]

 & R^{s-1}\bDel\bdel\bd\bdel^* T^s & \a\b= 00, 1\leq s, 1\leq\k\leq (n-1) \\[4pt]
\hline 						
\end{array}
\end{equation}
\end{thm} 

\begin{proof}
By Eq.\ \eqref{Eq:fSN}, we have
 $\fS^\d=\cF\circ(\Id+N)$ with $N$ a $\ce(p,q)$-invariant operator 
from $\sT^\d_{k,\k,s;\a\b}$ to $\bigoplus_{1\leq \ell<k}\sT^\d_{k-\ell,\k+2\ell}$. 
By Propositions \ref{GgEquiv} and \ref{ce-InvDiff}, the operator
$N$ is in fact a $\CE(p,q)$-invariant differential operator. 
Using Formula \eqref{cQeucl}, we deduce that  $N$ 
 is a linear combination of the following linearly independent operators
$$
\bd\bdel, \quad R^s\bDel\bdel D T^s, \quad R^s\bDel\bdel\bd\bDel^* T^s, \quad R^{s-1}\bDel\bdel G_0T^s, \quad R^{s-1}\bDel\bdel \bd \bdel^*_0 T^s.
$$  

The $\fkg$-equivariance of $\fS^\d$ reads as $\fS^\d\circ\Lt=\Ld\circ\fS^\d$ for all $X\in\fkg$.
By the Formula \eqref{EqLtLd}, relating the $\fkg$-action on $\sS^\d$ and $\sT^\d$, the $\fkg$-equivariance of~$\fS^\d$ is equivalent to 
\begin{equation}\label{EquivRelationSuper}
[N,\bbL^\d_{\bar{X}_i}]=-2\frac{\hbar}{\bi}\xi_i\bdel,
\end{equation}
for all $i=1,\ldots,n$. The actions of the conformal inversion $\bar{X}_i$ on the operators $G,D,\bd,\bd^*$ are given in \eqref{Commutateur}, while the operators $R^sT^s$ and $\bDel\bdel$ are conformally invariant. We deduce that, on $\sT^\d_{k,\k,s}$,
\begin{eqnarray*}
[\bd\bdel,\bbL^\d_{\bar{X}_i}]&=&2(k+\k+1-n\d)\xi_i\bdel+\frac{2}{\rho_{k,s}}R^s\bDel\bdel\partial_{p_i}T^s+\frac{4s}{\rho_{k,s}}R^{s-1}\bDel\bdel p_iT^s,\\
\left[R^s \bDel\bdel DT^s,\bbL^\d_{\bar{X}_i}\right]&=& 2(2(k-s-1)+n(1-\d))R^s\bDel\bdel\partial_{p_i}T^s-2R^s\bDel\bdel\xi_i\bDel^* T^s,\\
\left[ R^s\bDel\bdel\bd\bDel^* T^s ,\bbL^\d_{\bar{X}_i}\right]&=& 2(k+\k-n\d)R^s\bDel\bdel\xi_i\bDel^* T^s,\\
\left[R^{s-1}\bDel\bdel G_0 T^s,\bbL^\d_{\bar{X}_i}\right]&=&2(2s-n\d)R^{s-1}\bDel\bdel\Pi_0 p_iT^s-2R^{s-1}\bDel\bdel\xi_i\bdel^*_0 T^s,\\
\left[R^{s-1}\bDel\bdel \bd \bdel^*_0 T^s,\bbL^\d_{\bar{X}_i}\right]&=& 2(k+\k-n\d)R^{s-1}\bDel\bdel\xi_i\bdel^*_0 T^s.
\end{eqnarray*}
The formula giving $\fS^\d$ is then deduced from the relation $R\partial_{p_i}T^s+2sp_iT^s=2s\Pi_0 p_i T^s+(a_{k,s})^{-1}R\partial_{p_i}T^s$. The remaining statements of the theorem easily follow.
\end{proof}

We provide an alternative formula for the conformally equivariant superization which holds on the whole tensor module $\sT^\d$.

\begin{prop}
Let $\d\in\bbR\setminus I^\fS$. On $\sT^\d$, the conformally equivariant superization is given by
\begin{equation}\label{EqrS} 
\fS^\d =\cF\circ\bigg(\Id  +\frac{\hbar}{\bi}C_{\bd}\Big[\bd\bdel+\bDel\bdel(n\d \mathbb{D}-\mathbb{G}T)C_{G}\Big]\bigg),
\end{equation}
where 
\begin{equation}\label{Coefbis_S}
\left\{\begin{array}{l}
(C_{\bd})^{-1}=-(\cE+\Sigma -n\d),\\ [6pt]
(C_{G})^{-1}=RT-n\d\big(2(\cE-1)+n(1-\d)\big),
\end{array}\right.
\qquad \qquad
\left\{\begin{array}{l}
\mathbb{D}=D+\frac{1}{\cE+\Sigma+1-n\d}\bd\bDel^*,\\ [6pt]
\mathbb{G}=G+\frac{1}{\cE+\Sigma+1-n\d}\bd\bdel^*.
\end{array}\right.
\end{equation}
\end{prop}

\begin{proof}
By Eq.\ \eqref{Commutateur} and the conformal invariance of $\bDel\bdel$, we get
the following equalities of operators acting on $\sT^\d$,
\begin{eqnarray*}
[\bd\bdel,\bbL^\d_{\bar{X}_i}]&=&2(\cE+\Sigma-n\d)\xi_i\bdel+\bDel\bdel\partial_{p_i},\\[3pt]
[\bDel\bdel\mathbb{D},\bbL^\d_{\bar{X}_i}]&=&2\bDel\bdel\big(\partial_{p_i}(2(\cE-1)+n(1-\d))-p_i T \big),\\[3pt]
[\bDel\bdel\mathbb{G}T,\bbL^\d_{\bar{X}_i}]&=&2\bDel\bdel\big(-n\d p_iT+\partial_{p_i}RT\big).
\end{eqnarray*}
As a consequence, the operator
$$
N=\frac{\hbar}{\bi}C_{\bd}\Big[\bd\bdel+\bDel\bdel(n\d \mathbb{D}-\mathbb{G}T)C_{G}\Big]
$$
satisfies the equivariance Eq.\ \eqref{EquivRelationSuper}.
By uniqueness of the map $\fS^\d$, we deduce that  $\fS^\d=\cF\circ(\Id+N)$.
\end{proof}

Concerning the conformally equivariant quantization, a general formula is out of reach. 
We restrict ourselves to $\sS^\d_{\leq 1}$, the space of symbols of $p$-degree $0$ or $1$. 

\begin{thm}\label{thm_equi_adaptees}
Let $\l,\m\in\bbR$ such that $\m-\l=\d$.
For ${n\d\neq 1,\ldots,n+1}$, there exists a unique conformally equivariant quantization $\cQlm:\sS^\d_{\leq 1}\rightarrow \sD^{\l,\m}_1$, which reads as
\begin{equation}\label{cQConfInvIso}
\cQlm=\cN \circ\cF\circ \left(\Id + \frac{\hbar}{\bi}\Big[c_D(\Sigma)\,D + c_{\bdel}(\Sigma) \,\bd^{*} \bdel +c_{\bd}(\Sigma)\, \bd \bDel^* + c_{*}(\Sigma)\, \bd^{*} \bDel^*\Big]\right)\circ\cF^{-1},  
\end{equation} 
the coefficients depending on the odd Euler operator $\Sigma=\xi^i\partial_{\xi^i}$ as follows
\begin{equation}\label{coef_generiqueQ} 
\left\{\begin{array}{l}
c_D(\Sigma)=\frac{2n\l+1}{2n(1-\d)+2},\\ [6pt]
c_{\bdel}(\Sigma) =\frac{n(1-\d-2\l)}{(\Sigma-n(1-\d))(2n(1-\d)+2)}, \\[6pt] 
c_{\bd}(\Sigma) =-\frac{n(1-\d-2\l)}{(\Sigma-n\d)(2n(1-\d)+2)},\\ [6pt]
c_{*}(\Sigma) =\frac{1}{4(\Sigma-n(1-\d))}.
\end{array}\right.
\end{equation} 
The remaining critical cases, $n\d=1,\ldots n+1$, correspond to the existence of a $\fkg$-invariant operator on certain subspaces $\sT^\d_{1,\k,0;\a\b}$ as below
\begin{equation}\label{tableInvQuant}
\setlength{\extrarowheight}{4pt}
\begin{array}{|c||c|c|c|c|}
\hline
\d & \fkg-\text{invariant operator}   & \text{source space $\sT^\d_{1,\k,0;\a\b}$} & \l \text{ s.t. } \cQ^{\l,\l+\d} \text{ exists} \\[4pt]
\hline
\frac{1}{n} & \bd\bDel^* & \k=1, \a=1& \l=\frac{n-1}{2n} \\[4pt]
 & \bd^*\bdel & \k=n-1, \a\b=01 & \l=\frac{n-1}{2n} \\[4pt]
\hline
\frac{\ell}{n},  & \bd\bDel^*& \k=\ell, \a=1& none\\[4pt]
 2\leq \ell\leq n & \bd^*\bdel & \k=n-\ell, \a\b=01& none\\[4pt]
  & \bd^*\bDel^* & \k=n-\ell+2, \a\b=10& none\\[4pt]
\hline
\frac{n+1}{n} & D+\frac{1}{\Sigma-n-1}\bd\bDel^*-\frac{1}{\Sigma+1}\bd^*\bdel & \k=0,\ldots,n & \l=\frac{-1}{2n}  \\[4pt]
\hline 						
\end{array}
\end{equation}
For the two critical values $\d=\frac{1}{n},\frac{n+1}{n}$, the quantization $\cQ^{\l,\l+\d}$ exists for a unique value of $\l$ and is given by, respectively,
\begin{eqnarray*}
\cQ^{\frac{n-1}{2n},\frac{n+1}{2n}}&=&\cN \circ\cF\circ \left(\Id + \frac{\hbar}{\bi}\Big[ c_D(\Sigma)\,D + c_{*}(\Sigma)\, \bd^{*} \bDel^*\Big]\right)\circ\cF^{-1},  \\ [8pt]
\cQ^{\frac{-1}{2n},\frac{2n+1}{2n}} &=&\cN \circ\cF\circ \left(\Id + \frac{\hbar}{\bi}\Big[ c_{\bdel}(\Sigma) \,\bd^{*} \bdel +c_{\bd}(\Sigma)\, \bd \bDel^* + c_{*}(\Sigma)\, \bd^{*} \bDel^*\Big]\right)\circ\cF^{-1}.  
\end{eqnarray*}
These two maps are not unique. On the $\fkg$-module $\fS^\d(\sT^\d_{1,\k,0;\a\b})$, they can be precomposed by  
$\Id+c\, \fS^\d(N)$, with $c\in\bbR$ and $N$ a $\fkg$-invariant operator as given in the list above.
\end{thm}

\begin{proof}
By Eq.\ \eqref{Eq:fSN}, we have $\cQlm=\cN\circ\cF\circ(\Id+N)$ 
with $N$ a $\ce(p,q)$-invariant operator from $\sT^\d_{1,\k,0;\a\b}$ 
to $\bigoplus_{ \k'<\k+2}\sT^\d_{0,\k'}$. 
By Propositions \ref{GgEquiv} and \ref{ce-InvDiff}, the operator
$N$ is in fact a $\CE(p,q)$-invariant differential operator. 
Using Formula \eqref{cQeucl}, we deduce that  $N$ 
 is a linear combination of the following linearly independent operators,
$$
D,\quad \bd^*\bdel, \quad \bd\bDel^*,\quad \bd^*\bDel^*.
$$

The $\fkg$-equivariance of $\cQ^{\l,\m}$ reads as $\cQlm\circ\Ld=\LD\circ\cQlm$ for all $X\in\fkg$.  
By the Formul{\ae} \eqref{EqLtLd}--\eqref{EqLdLD}, relating the $\fkg$-action on $\sT^\d$, $\sS^\d$ and $\Dslm$, the $\fkg$-equivariance of $\cQlm$ is equivalent to
\begin{equation}\label{GEquiN}
[N,\bbL^\d_{\bar{X}_i}]=\frac{\hbar}{\bi}\big(2n\l\partial_{p_i}+\xi^i\bDel^*-\bdel\partial_{\xi_i}+\half\partial_{\xi_i}\bDel^*\big),
\end{equation}
for all $i=1,\ldots,n$.
From the Relations  \eqref{Commutateur} in appendix we deduce the following equalities in $\sT^\d_{1,\k}$, 
\begin{eqnarray*}
\left[\bd^*\bdel,\bbL_{\bar{X}_i}^\d\right]&=&2(\Sigma-n(1-\d)) (\bdel\partial_{\xi^i}-\partial_{p^i}),\\
\left[\bd\bDel^*,\bbL_{\bar{X}_i}^\d\right]
&=&2 (\Sigma-n\d)\xi_i\bDel^*,\\
\left[\bd^*\bDel^*,\bbL_{\bar{X}_i}^\d\right]&=&2(-\Sigma+n(1-\d)) \partial_{\xi^i}\bDel^*,\\
\left[D,\bbL_{\bar{X}_i}^\d\right]&=& 2\bdel\partial_{\xi^i}-2\xi_i\bDel^*+2n(1-\d)\partial_{p^i}.
\end{eqnarray*}
By  substitution in the $\fkg$-equivariance Condition \eqref{GEquiN}, this determines a unique operator $N$, as specified by \eqref{coef_generiqueQ}. The remaining statements of the theorem easily follow.
\end{proof}

\begin{rmk}
Via the conformally equivariant superization $\fS^\d$, the $\fkg$-invariant operators listed in \eqref{tableInvQuant} 
can be transported as $\fkg$-invariant operators acting on $\sS^\d$. They provide obstruction to the existence and uniqueness
of the conformally equivariant quantization, in accordance with Theorem \ref{ExistenceSuper}.
%Attention: $\fS^\d$ n'existe pas sur $\sT^{\frac{n+1}{2}}_{1,\frac{n-1}{2},0;01}\cong\Om^{\frac{n}{2}}(M)$, car $\bd$ est $\fkg$-invariant. Pour transporter $\bd^*$, il faut décomposer en $1/2$-formes paires et impaires, de sorte que les deux opérateurs sont invariants sur deux espaces distincts.
\end{rmk}

\begin{rmk}
After suitable restriction, corestriction and pull-back to 
$\ker T\cap\ker\bdel\cap\ker\bDel^*$ (via the commutative Diagram \eqref{Diag:00}),
each $\fkg$-invariant operator in Tables \eqref{tableInvSuper}  and \eqref{tableInvQuant} 
corresponds to one of the $\fkg$-invariant operators in Theorem \ref{thmConfInv}.
 They are respectively $\bd_{\bf 0},G_{\bf 0},D_{\bf 0},\bd_{\bf 0}$ for Table \eqref{tableInvSuper}
and $\bd_{\bf 0},\bd^*_{\bf 0},\bd_{\bf 0},\bd^*_{\bf 0},\bd^*_{\bf 0},D_{\bf 0}$  for Table \eqref{tableInvQuant}.
\end{rmk}

\begin{rmk}
By uniqueness of the conformally equivariant superization and quantization, the obtained formul{\ae} for $\fS^\d$ and $\cQlm$ are globally defined on arbitrary conformally flat manifolds $(M,\mg)$. Nevertheless, their building blocks are only locally defined.  
\end{rmk}

%%%%%%%%%%%%%%%%%%%%%%%%%%%%%%%%%%%%
\subsection{Covariant formul{\ae} and the curved case}
%%%%%%%%%%%%%%%%%%%%%%%%%%%%%%%%%%%%

Starting with a pseudo-Riemannian spin manifold $(M,\mg)$ endowed with the Levi-Civita connection, 
we get covariant derivations on all associated bundles to the principal bundle of spin frames. 
For $X\in\Vect(M)$, we denote by $\nabla^\l_X$ the one acting on $\Ga(S\otimes|\Lambda|^\l)$, 
and by $\partial^\nabla_X$ the one acting on $\sT^\d=\Ga(\bbT^\d)$. 
The latter covariant derivative can be interpreted as the horizontal covariant derivative on the supercotangent bundle $\cM$. It allows to generalize the $13$ invariant operators introduced in Proposition \ref{InvIso} as global operators over any pseudo-Riemannian manifold $(M,\mg)$, replacing the flat metric $\eta$ by $\mg$ and the derivatives $\partial_i$ by the covariant ones $\partial^\nabla_i$. Their commutation relations remain the same than in the flat case, see \eqref{CommRel}, except between the five ones containing covariant derivatives.  Accordingly, we denote them with a superscript $\nabla$,  
\begin{equation}\label{Op:nabla}
G^\nabla=\mg^{ij}p_i\partial^\nabla_j, \quad D^\nabla=\partial_{p_i}\partial^\nabla_i, \quad L^\nabla=\partial^\nabla_i \mg^{ij}\partial^\nabla_j, \quad \bd^\nabla=\xi^i\partial^\nabla_i, \quad \bd^{*\nabla}=\mg^{ij}\partial_{\xi^i}\partial^\nabla_j.
\end{equation}
This allows for the following definitions.

\begin{defi}
Let $(M,\mg)$ be a pseudo-Riemannian spin manifold of even dimension, $P\in\sT^\d$, $P_1+P_0\in\sS^\d_{\leq 1}$ 
(the index standing for the $p$-degree) and $\l,\m\in\bbR$ such that $\m-\l=\d$. 
Using the coefficients introduced in \eqref{Coefbis_S} and \eqref{coef_generiqueQ}, we define the applications
\begin{equation} \label{cSCovariant}
\fS^\d_\nabla(P) := |\vol_\mg|^{\frac{\Sigma}{n}}\left( P+\frac{\hbar}{\bi}C_{\bd}\Big[\bd^\nabla\bdel+\bDel\bdel(n\d \mathbb{D^\nabla}-\mathbb{G^\nabla}T)C_{G}\Big] P\right)
\end{equation}
and%with straightforward definitions for $\mathbb{D}^\nabla$ and $\mathbb{G}^\nabla$.
\begin{multline}
\cQ^{\l,\m}_\nabla(P_1+P_0):=\frac{\hbar}{\bi} \, \left(\g\Big(\frac{1}{(\sqrt{2})^{\Sigma}}(\partial_{p_i}P_1)\Big)\nabla^\l_i\right)+\g(P_0)\\  
+\frac{\hbar}{\bi}\, \g\left(\frac{1}{(\sqrt{2})^{\Sigma}}\Big(c_D(\Sigma)\, D^\nabla P_1 + c_{\bdel}(\Sigma)\, \bd^{*\nabla}  \bdel P_1 +c_{\bd}(\Sigma)\, \bd^\nabla \bDel^* P_1 + c_{*}(\Sigma)\, \bd^{*\nabla}  \bDel^* P_1\Big)\right).\\ \label{cQCovariant}
\end{multline}
\end{defi}

Next result shows that these maps extend the maps $\fS^\d$ and $\cQlm$ to arbitrary pseudo-Riemannian spin manifolds.

\begin{prop}\label{Curved-Flat}
The maps $\fS^\d_\nabla$ and $\cQ^{\l,\m}_\nabla$ coincide with, respectively, the conformally equivariant superization and quantization if $(M,\mg)$ is conformally flat. 
\end{prop}

\begin{proof}
Let $(M,\mg)$ be a conformally flat manifold and $x$ a point in $M$. 
We prove that $\fS^\d_\nabla(P)$ and $\fS^\d(P)$ are equal at $x$. 
As the metric $\mg$ is conformally flat, there exists conformal coordinates in a neighborhood of $x$,
 such that $\mg_{ij}=F\eta_{ij}$ for $F$ a positive function and $\eta$ the flat metric.
Moreover, up to a conformal change of coordinates we can assume that all the first derivatives of $F$ at $x$ vanish. This is a classical result, which can be derived from the infinitesimal action of inversions $L_{\bar{X}_i}\eta=-4x_i\eta$. Using such a coordinate system, we can write down both $\fS^\d_\nabla(P)$ and $\fS^\d(P)$ at $x$. They coincide  since the covariant Formula \eqref{cSCovariant} involves only first order derivations. 
The same argument applies for the quantization Formula \eqref{cQCovariant}.
\end{proof}

The conformally equivariant quantization of cotangent bundles was 
extended to every pseudo-Riemannian manifold in a conformally invariant fashion  \cite{DOv01,MRa09,Sil09}. 
The situation here is more involved,  
only the composition $\cQlm\circ\fS^\d$ admits a conformally invariant formula. 

\begin{thm}\label{cQlmConfInvariant}
Let $(M,\mg)$ be a pseudo-Riemannian spin manifold, $P\in\sT_1^\d$ and $\l,\m\in\bbR$ such that $\m-\l=\d$.
The following map
\begin{multline*} 
\mathsf{Q}^{\l,\m}_\nabla(P)=\frac{\hbar}{\bi} \, \g\left(\frac{|\vol_\mg|^{\frac{\Sigma}{n}}}{(\sqrt{2})^{\Sigma}}\,(\partial_{p_i}P)\right)\nabla^\l_i\\ 
+\frac{\hbar}{\bi} \,\g\left(\frac{|\vol_\mg|^{\frac{\Sigma}{n}}}{(\sqrt{2})^{\Sigma}}\Big(c_0\,\bd^\nabla\bdel P+c_1\,D^\nabla P + c_2\, \bd^{*\nabla}  \bdel P +c_3\,\bd^\nabla \bDel^* P + c_4\, \bd^{*\nabla}  \bDel^* P\Big)\right),
\end{multline*}
is conformally invariant, i.e.\ depends only on the conformal class of the metric $\mg$, if and only if the real coefficients $c_0,c_1,c_2,c_3,c_4$ are such that $\mathsf{Q}^{\l,\m}_\nabla$ is equal to the composition $\cQ_\nabla^{\l,\m}\circ\fS^\d_\nabla$. Neither $\fS^\d_\nabla$ nor $\cQ_\nabla^{\l,\m}$ are conformally invariant.
\end{thm}

\begin{proof}
We suppose that $\mg$ and $\hat{\mg}$ are two metrics conformally related through $\hat{\mg}=F\mg$, with $F$ a positive function on $M$. The Christoffel symbols of their Levi-Civita connections satisfy
$$
\hat{\Ga}_{ij}^k- \Ga_{ij}^k=\frac{1}{2F}\left(F_i\d^k_j+F_j\d^k_i-F^k\mg_{ij} \right),
$$
where $F_i=\partial_i F$ and $F^k=\mg^{ik}\partial_i F$. 
%From the usual expressions of the covariant derivative on spinor and tensor bundles in terms of Christoffel symbols, 
We deduce that
\begin{eqnarray*}
\hat{\nabla}^\l_j-\nabla^\l_j &=& -\frac{1}{8F} [\g_j,\g^i]F_i-\frac{n\l}{2F}F_j ,\\
\partial^{\hat{\nabla}}_i- \partial^\nabla_i &=&\frac{1}{2F}\left(F_i\cE+F_k p_i\partial_{p_k}- F^jp_j\partial_{p^i}\right) - \frac{1}{2F}\left(F_k \xi^k\partial_{\xi^i}- F^j\xi_i\partial_{\xi^j}\right) - \frac{n\d}{2F}F_i ,
\end{eqnarray*}
where $\nabla^\l$ is the covariant derivative acting on $\Ga(S\otimes|\Lambda|^\l)$
and $\partial^\nabla$ the one acting on $\sT^\d$. 
From the conformal invariance of $|\vol_\mg|^{\frac{\Sigma}{n}}\g$, and of $\bdel$, $\bDel^*$ as operators from $\sT^\d$ to $\sT^{\d-1/n}$, we deduce the expression of each of the six terms involved in the operator
 $\mathsf{Q}^{\l,\m}_{\hat{\nabla}}-\mathsf{Q}^{\l,\m}_{\nabla}$: 
\begin{eqnarray*}
|\vol_\mg|^{\frac{\Sigma}{n}}\g(\partial_{p_j}P)(\hat{\nabla}^\l_j -\nabla^\l_j)&=&-|\vol_\mg|^{\frac{\Sigma}{n}}\g\left( \big[\bdel\xi^i-\half\xi^i\bDel^*+\half\bdel\partial_{\xi_i}+\frac{1}{4}\bDel^*\partial_{\xi_i}+n\l\partial_{p_i}\big] P\right)\frac{F_i}{2F},\\[6pt]
c_0\,\left(\bd^{\hat{\nabla}}-\bd^\nabla\right)\bdel 
&=& c_0\left((\Sigma-n\d)\xi^i\bdel\right)\frac{F_i}{2F},\\
c_1\left(D^{\hat{\nabla}}-D^\nabla\right)
&=&
c_1\left( \bdel\partial_{\xi_i}- \xi^i\bDel^* +n(1-\d)\partial_{p_i}\right)\frac{F_i}{2F},\\
c_2\,\left(\bd^{*\hat{\nabla}}-\bd^{*\nabla} \right)\bdel 
&=&
c_2\left((\Sigma-n(1-\d))\bdel\partial_{\xi_i}\right)\frac{F_i}{2F},\\
c_3\,\left(\bd^{\hat{\nabla}}-\bd^\nabla\right)\bDel^* 
&=&
c_3\left((\Sigma-n\d)\xi^i\bDel^*\right)\frac{F_i}{2F},\\
c_4\,\left(\bd^{*\hat{\nabla}}-\bd^{*\nabla} \right)\bDel^* 
&=& c_4\left((\Sigma-n(1-\d))\bDel^*\partial_{\xi^i}\right)\frac{F_i}{2F}.
\end{eqnarray*}
The first equality holds for all $P\in\sT^\d_1$ whereas the other ones 
should be understood as equalities of operators acting on $\sT^\d_1$.
The equality $\mathsf{Q}^{\l,\m}_{\hat{\nabla}}(P)-\mathsf{Q}^{\l,\m}_{\nabla}(P)=0$, for all $P\in\sT^\d_1$, is then equivalent to a linear system in the five coefficients entering the definition of $\mathsf{Q}^{\l,\m}_{\nabla}$. Solving this system leads exactly to $\mathsf{Q}^{\l,\m}_\nabla=\cQ^{\l,\m}_\nabla\circ\fS^\d_\nabla$, as wanted.

Since $c_0=c_{\bd}(\Sigma)=\frac{-1}{\Sigma-n\d}$, the preceding computations lead to
$$
\fS^\d_{\hat{\nabla}}(P)=F^{\Sigma/2}\left(\fS^\d_\nabla(P)-\frac{\hbar}{\bi}\frac{F_j}{2F}\xi^j\bdel P\right),
$$
and then neither $\fS^\d_\nabla$ nor $\cQ^{\l,\m}_\nabla$ are conformally invariant. 
\end{proof}

\begin{rmk}
Let $(M,\mg)$ be a conformally flat manifold and $A$ a first order differential operator 
acting on the tensor bundle $\cS TM\otimes\Lambda T^*M$. Then, $\fkg$-equivariance of $A$ implies its conformal invariance, whenever the $\fkg$-actions are the natural ones on tensors \cite{Feg76}. 
But the Hamiltonian action is not of this type, as seen in \eqref{EqLtLd}. 
This explains the lack of conformal invariance of the maps $\fS^\d$ and $\cQlm$.
\end{rmk}

%%%%%%%%%%%%%%%%%%%%%%%%%%%%%%%%%%%%%%%%%%%%%%%%%
%%%%%%%%%%%%%%%%%%%%%%%%%%%%%%%%%%%%%%%%%%%%%%%%%
\section{Symmetries of spinning particles}
%%%%%%%%%%%%%%%%%%%%%%%%%%%%%%%%%%%%%%%%%%%%%%%%%
%%%%%%%%%%%%%%%%%%%%%%%%%%%%%%%%%%%%%%%%%%%%%%%%%

We classify all the symmetries of free massless spinning particles over conformally flat manifolds, both in the classical and quantum cases. They arise via the conformally equivariant superization and quantization of conformal Killing hook tensors. We introduce 
the latters over an arbitrary pseudo-Riemannian manifold. 
%These maps will allow us to obtain all conserved quantities from symmetries of the underlying manifold $(M,\mg)$, when it is conformally flat. We also deal with general manifolds in the classical case, the quantum case being left for future studies.

%%%%%%%%%%%%%%%%%%%%%%%%%%%%%%%%%%%%%%%%%%%%%%%%%
\subsection{Reminder on the non-spinning case}
%%%%%%%%%%%%%%%%%%%%%%%%%%%%%%%%%%%%%%%%%%%%%%%%%

A classical particle on a pseudo-Riemannian manifold $(M,\mg)$ admits the cotangent bundle $T^*M$ as phase space,
endowed with its canonical Poisson bracket $\{\cdot,\cdot\}$. In the free case, its motion is given by 
the Hamiltonian flow of $R=\mg^{ij}p_ip_j$, which projects onto the geodesic flow of $(M,\mg)$. 
The symmetries, or conserved quantities, are the elements $K\in\Pol(T^*M)$ such that $\{R,K\}=0$,
i.e. symmetric Killing tensors (here and thereafter we freely use the identification $\Pol(T^*M)\cong\Ga(\cS TM)$).
If the particle is in addition massless, it moves along the null cone, characterized by $R=0$. 
Then, the symmetries are the elements $K\in\Pol(T^*M)$ such that $\{R,K\}=0$ on the null cone. 
This means $\{R,K\}\in (R)$ with $(R)\subset\Pol(T^*M)$ the ideal generated by $R$. 
Such symmetries form a subalgebra $\hat{\cK}\leq\Pol(T^*M)$, which contains $(R)$.
Moreover the elements in $(R)$ are considered as trivial symmetries since they vanish on the null cone.
Hence, relevant symmetries are given by elements in the algebra $\cK=\hat{\cK}/(R)$
or equivalently by traceless tensors in $\hat{\cK}$. 
They are  characterized as follows.

\begin{prop}\label{PropKilling}
Let $K$ be a traceless symmetric $k$-tensor, with $k\in\bbN$. 
Then,  the three following conditions are equivalent
$$
\left\{\begin{array}{l}
\{R,K\}\in(R), \\ 
\Pi_{0}G^\nabla K=0,\\
\Pi_{ 0}\nabla_{(i_0}K_{i_1\cdots i_k)}=0,
\end{array}\right.
$$
where $\Pi_0$ denotes the projection on traceless tensors and the round brackets denote symmetrization.
If $K$ satisfies one of the above condition, it is called a conformal Killing symmetric tensor. 
\end{prop} 

%The sapce of conformal Killing tensors identifies with $\cK_0$.

Roughly speaking, a quantum massless free particle is given by a function
in the kernel of the conformal Laplacian 
$\Delta=\nabla_{\! i\,}\mg^{ij}\nabla_{\! j\,}-\frac{n-2}{4(n-1)}\mathrm{Scal}$, where $\mathrm{Scal}$ denotes the scalar curvature. The latter is a conformally invariant operator if considered as an element of the space $\Dlm:=\D(M;|\Lambda|^\l,|\Lambda|^\m)$ for $\l=\frac{n-2}{2n}$, $\m=\frac{n+2}{2n}$. Following \cite{Eas05}, we introduce higher symmetries of $\Delta$.

\begin{defi}
Let $\l=\frac{n-2}{2n}$ and  $\m=\frac{n+2}{2n}$.
A higher symmetry of $\Delta$ is a differential operator $D_1\in\D^{\l,\l}$, such that $\Delta D_1=D_2\Delta$, for some $D_2\in\D^{\m,\m}$.
\end{defi} 

If $D_1=A\Delta$, with $A\in\D^{\m,\l}$, then $D_1$ is a higher symmetry of $\Delta$, called a trivial higher symmetry.
The higher symmetries preserve the kernel of $\Delta$ and form a subalgebra $\hat{\cA}\leq\D^{\l,\l}$,
whereas the trivial higher symmetries act trivially on $\ker\Delta$ and form the ideal 
$(\Delta)=\{A\Delta, \text{ s.t. } A\in\D^{\m,\l}\}$.
Hence, we do not distinguish between two higher symmetries differing by a trivial one, 
and consider the quotient algebra $\cA=\hat{\cA}/(\Delta)$  of equivalence classes of higher symmetries.
They are quantum analogs of the conformal Killing symmetric tensors.

\begin{thm}\label{thm:symLaplacian}\cite{Eas05,Mic11b}
Let $(M,\mg)$ be a conformally flat manifold. The conformally equivariant quantization 
induces an isomorphism of $\fkg$-modules $\cQ^{\l,\l}:\hat{\cK}\rightarrow \hat{\cA}$, 
such that $\cQ^{\l,\l}((R))=(\Delta)$.
It descends as an isomorphism of $\fkg$-modules, 
$\cQ^{\l,\l}:\cK\rightarrow \cA$, 
mapping conformal Killing symmetric tensors to (equivalence classes of) higher symmetries of $\Delta$.
\end{thm}

%%%%%%%%%%%%%%%%%%%%%%%%%%%%%%%%%%%%%%%%%%%%%%%%%
\subsection{Conformal Killing hook tensors}
%%%%%%%%%%%%%%%%%%%%%%%%%%%%%%%%%%%%%%%%%%%%%%%%%
 We call hook tensors the tensors whose symmetry is described by a hook Young diagram of the following type:
 \raisebox{-17pt}{$
\begin{picture}(40,40)
\put(0,0){\line(0,1){40}}
\put(10,0){\line(0,1){40}}
\put(0,40){\line(1,0){40}}
\put(0,30){\line(1,0){40}}
\put(30,30){\line(0,1){10}}
\put(40,30){\line(0,1){10}}
\put(20,35){\makebox(0,0){$\cdots$}}
\put(0,0){\line(1,0){10}}
\put(0,10){\line(1,0){10}}
\put(5,25){\makebox(0,0){$\vdots$}}
\end{picture}
$}.
They correspond to elements in the tensor algebra $\sT^0$ and are referred as $(k,\k)$-tensors if they lie in $\sT^0_{k,\k}$. Thus, $(k,0)$-tensors are symmetric and $(0,\k)$-tensors are skew-symmetric. We define in this section the notion of conformal Killing hook tensors, so that it generalizes both conformal Killing symmetric tensors and conformal Killing forms, defined below.

\begin{defi}
Let $1\leq \k\leq n$.
A differential $\k$-form $\eta$ is conformal Killing, if it satisfies
\[%\begin{equation}
\nabla_X\eta=\frac{1}{\k+1}\left\langle X,\mathbf{d}\eta \right\rangle +\frac{1}{n-\k+1} X^\flat\wedge \mathbf{d}^*\eta,
\]
for all $X\in\Vect(M)$. Here, $\mathbf{d}^*$ and $\mathbf{d}$ denote the de Rham (co-)differentials, and $X^\flat$ is the dual $1$-form of $X$ through the metric $\mg$. 
\end{defi}

We would like to characterize the space of conformal Killing forms as the kernel of a conformally invariant operator. 
We use the invariant operators introduced in \eqref{EqInvIsoSpo} and \eqref{Op:nabla}, as well as the conformally invariant projection $\Pi_{\bf 0}:\sT^\d\rightarrow\bigoplus_{k,\k}\sT^\d_{k,\k,0;00}$, whose image is equal to
$\ker T\cap\ker \bdel\cap\ker\bDel^*$.

\begin{prop}\label{Prop:CKY}
The operator $\Pi_{\bf 0}\circ G^\nabla:\sT^{-\frac{1}{n}}\rightarrow\sT^{\frac{1}{n}}$ 
is conformally invariant and, on $\sT^{-\frac{1}{n}}_{k,\k,s;\a\b}$, it is given by
\begin{equation}\label{OpCK}
\Pi_{\bf 0}\circ G^\nabla=
\begin{cases}
G^\nabla-\frac{1}{\cE+\Sigma}\bdel^*\bd^\nabla-\frac{1}{n+\cE-\Sigma}\bDel\bd^{*\nabla} -\frac{1}{n+2\cE-4}RD^\nabla, 
&\text{if }  s=0, \a\b=00,\\
0, & \text{else}.
\end{cases}
\end{equation}
On $\sT^{-\frac{1}{n}}_{0,\k,0;00}$, with $1\leq \k\leq n$, its kernel is the space of conformal Killing $\k$-forms.
\end{prop}

\begin{proof}
The operator $\Pi_{\bf 0}\circ G^\nabla$ is the composition of the Levi-Civita covariant derivative 
with the projection on an irreducible homogeneous subbundle.
By Theorem~\ref{thmConfInv}, it is $\fkg$-invariant in the conformally flat case, 
and  by Fegan's work \cite{Feg76}, it is invariant by a conformal change of metric over a general manifold $(M,\mg)$.

The operator $G^\nabla$ commutes with $R$, $\bdel^*$ and $\bDel$, so that the vanishing of $\Pi_{\bf 0}\circ G^\nabla$ on $\sT^{-\frac{1}{n}}_{k,\k,s;\a\b}$ if $s\neq 0$ or $\a\b\neq 00$ is obvious. In the remaining case, the expression of $\Pi_{\bf 0}\circ G^\nabla$ follows from the vanishing of its composition with $T$, $\bdel$ and $\bDel^*$. The coefficients are computed thanks to the Table \eqref{CommRel} of commutation relations. 

Let $\eta\in\sT^{-\frac{1}{n}}_{0,\k,0;00}$, and $X\in\Vect(M)$. We get
$$
\mg_{ij}X^i\partial_{p_j}(\Pi_{\bf 0}\circ G^\nabla \eta)=\nabla_X\eta-\frac{1}{\cE+\Sigma}X^i\partial_{\xi^i}\bd^\nabla \eta-\frac{1}{n+\cE-\Sigma}\mg_{ij}X^i\xi^j\bd^{*\nabla}  \eta,
$$
where $\bd^{*\nabla} $ and $\bd^\nabla$ identify to de Rham (co-)differentials on differential forms. 
The right hand side vanishes for all $X\in\Vect(M)$ if and only if $\eta$ is a conformal Killing form. 
The left hand side vanishes for all $X\in\Vect(M)$ if and only $\eta$ lies in the kernel of $\Pi_{\bf 0}\circ G^\nabla$.
The conclusion follows.
\end{proof}

We introduce the conformally invariant projection $\Pi_{01}:\sT^\d\rightarrow\bigoplus_{k,\k}\sT^\d_{k,\k,0;01}$.
Recall that $\sT^\d_{0,0}\oplus\big(\bigoplus_{k,\k}\sT^\d_{k,\k,0;01}\big)=\ker \bdel^*\cap\left(\ker T\cap\ker\bDel^*\right)$.

\begin{defi}
If $(k,\k)\neq (0,0)$, a conformal Killing $(k,\k)$-tensor 
is an element in $\sT^0_{k,\k,0;01}\cap\ker(\Pi_{01}\circ G^\nabla)$.
By convention, a conformal Killing $(0,0)$-tensor is a constant function. 
We denote by $\cKh$ the space of all conformal Killing hook tensors.
\end{defi}

By Proposition \ref{PropKilling}, conformal Killing $(k,0)$-tensors are usual conformal Killing symmetric tensors of order $k$. 
Since $G^\nabla$ commutes with $\bdel^*$, the equality $\bdel^*\Pi_{\bf 0}=\Pi_{01}\bdel^*$ leads to $\bdel^*\Pi_{\bf 0}G^\nabla=\Pi_{01}G^\nabla\bdel^*$.
The equality $\ker\bdel^*\cap\ker\bdel=\{0\}$ on $\sT_{>0}^\d$ (see  Proposition \ref{Prop:proj})
and the Proposition \ref{Prop:CKY} show that conformal Killing $(1,\k-1)$-tensors identify to conformal Killing $\k$-forms through the operator $\bdel^*$. Therefore, the latter definition extends both notions. 

Since $\ker T\cap\ker\bDel^*\cong\sT^0/(\bDel,R)$,
the space $\cKh$ of conformal Killing hook tensors
identifies to the $\fkg$-submodule of $\sT^0/(\bDel,R)$
given by $\ker\bdel^*\cap\ker G^\nabla$.
As $\bdel^*$ and $G^\nabla$ are first order operators, 
$\cKh$ is also a subalgebra of $\sT^0/(\bDel,R)$.

%%%%%%%%%%%%%%%%%%%%%%%%%%%%%%%%%%%%%%%%%%%%%%%%%
\subsection{Pseudo-classical spinning particles}
%%%%%%%%%%%%%%%%%%%%%%%%%%%%%%%%%%%%%%%%%%%%%%%%%

The phase space of a spinning particle on a pseudo-Riemannian manifold $(M,\mg)$ 
is its supercotangent bundle $(\cM,\om)$, the spin being represented 
by a quadratic function in the Grassmann variables, $S=S_{ij}\xi^i\xi^j$. 
It automatically transforms in the right way under the action of the orthogonal group. 
This Hamiltonian model is equivalent to the well-known Lagrangian one developed in \cite{BMa77}, 
it is the classical counterpart of the quantum description of spinning particles as spinor fields \cite{Mic10a}.
In the free case, the equations of motion of the particle are again given by the Hamiltonian flow of $R$. 
Denoting by $x$ the parameterized trajectory, we get in particular
\begin{eqnarray*}
\nabla_{\dot{x}}\dot{x}^i&=&-\frac{1}{2}\mg^{il}(R^a_{~bkl}\xi_a\xi^b)\dot{x}^k,\\
\nabla_{\dot{x}}\xi^a&=&0,
\end{eqnarray*}
where $(R^a_{~bkl})$ denote the components of the Riemann tensor. 
The second equation shows that the spin is parallely transported, 
whereas the first one is an analog of Papapetrou's equation~\cite{Pap51}, 
which generalizes the geodesic equation for an extended object with spin in general relativity. 
The deviation from geodesic motion is due to a coupling between the curvature and the spin. 
Generally, the particle's spin is spacelike, that is expressed by $p_i\xi^i=0$.  
See \cite{Rav80,KCa11} for further informations.   

The principal Hamiltonian symbol of the Dirac operator is given by $\bDel=p_i\xi^i$, 
which Poisson squares to $R=\{\bDel,\bDel\}$. Thus, all symmetries 
of the Hamiltonian flow of $\bDel$ are also symmetries for the one of $R$. 
Following \cite{GRV93}, we call them supercharges.

\begin{defi}
A supercharge is an element $K$ of $\cO(\cM)$ such that $\{\bDel,K\}=0$. 
A conformal supercharge is an element $K$ of $\cO(\cM)$ such that 
$\{\bDel,K\}\in(\bDel,R)$, where $(\bDel,R)$ is the ideal in $\cO(\cM)$ generated by $\bDel$ and $R$. 
\end{defi}

Among the conformal supercharges stand all the elements in the ideal $(\bDel,R)$. 
They are considered as trivial  supercharges since they vanish if $\bDel=0=R$. 
Hence, we introduce the Poisson algebra of equivalence classes of conformal supercharges, defined by
$$
\cS\cC:=\{K\in\cO(\cM)|\{\bDel,K\}\in(\bDel,R)\}/(\bDel,R).
$$
It identifies to the reduced Poisson algebra resulting from the symplectic reduction of 
$\cM$ along the Hamiltonian flows of $\bDel$ and $R$. 
%Moreover, $\cS\cC$ is clearly a $\fkg$-submodule of $\sS^0$.

We classify below all the conformal supercharges in a close spirit to Proposition~\ref{PropKilling}.
This generalize \cite{GRV93,Tan95}, where supercharges are built from Killing forms. 

\begin{lem}
Let $(M,\mg)$ be a  pseudo-Riemannian manifold. 
For all $K\in\ker T\cap\ker\bDel^*$, 
the following equivalence holds
\begin{equation}\label{SuperConf}
\{\bDel,\fS^0_\nabla(K)\}\in(\bDel,R)
\quad\Longleftrightarrow\quad
%K\in (\bDel,R) \text{ or }
\begin{cases}
\bdel^* K=0,\\
\Pi_{01}G^\nabla K=\frac{\hbar}{\bi}\Pi_{01}(\bd^\nabla)^2 K.
\end{cases}
\end{equation}
\end{lem}

\begin{proof}
For $K\in\ker T\cap\ker\bDel^*$, we compute 
\begin{eqnarray}\nonumber
\{\bDel,\fS^0_\nabla(K)\}&=&\big(\bd^\nabla-\frac{\bi}{\hbar}\bdel^*\big)\circ\big(\Id-\frac{\hbar}{\bi}\frac{1}{\cE+\Sigma}\bd^\nabla\bdel\big) (K)+H,\\ \label{Compute:HSC}
&=& -\frac{\bi}{\hbar}\bdel^* K+\left(\bd^\nabla+\frac{1}{\cE+\Sigma}\bdel^*\bd^\nabla\bdel-\frac{\hbar}{\bi}\frac{1}{\cE+\Sigma}(\bd^\nabla)^2\bdel\right) (K)+H,\\ \nonumber
&=& -\frac{\bi}{\hbar}\bdel^* K+\frac{\bdel}{\cE+\Sigma}\Big(G^\nabla-\bd^\nabla\bdel^*-\frac{\hbar}{\bi}(\bd^\nabla)^2\Big)(K)+H,
\end{eqnarray}
with $H\in(\bDel,R)$. By Proposition \ref{Prop:proj}, we have $\im\bdel\cap\im\bdel^*=\{0\}$. 
Therefore, $\fS^0_\nabla(K)$ is a conformal supercharge if and only if  $\bdel^*K\in(\bDel,R)$ and $\bdel\Big(G^\nabla-\bd^\nabla\bdel^*-\frac{\hbar}{\bi}(\bd^\nabla)^2\Big)(K)\in(\bDel,R)$. 
Since $K\in\ker T\cap\ker\bDel^*$, this means $\bdel^*K=0$ and
$\Pi_{01}G^\nabla K=\frac{\hbar}{\bi}\Pi_{01}(\bd^\nabla)^2 K$. 
\end{proof}

\begin{prop}\label{Prop:KH-SC curved}
Let $(M,\mg)$ be a pseudo-Riemannian manifold. 
\begin{itemize}
\item If $S\in\sS^0_{k,\k}$ is a conformal supercharge 
then, up to trivial supercharges, its principal tensorial symbol 
is a conformal Killing hook tensor, i.e., $\ve_{k,\k}(S)\in(\bDel,R)\oplus\cKh$.
\item Let $K$ be a conformal Killing hook tensor. Then, $\fS^0_\nabla(K)$ is a conformal supercharge if and only if
$\Pi_{01}(\bd^\nabla)^2K=0$. The operator
$\Pi_{01}\circ (\bd^\nabla)^2=\Pi_{01}\circ\left(\xi^k\xi^lW^a_{~bkl} p_a\partial_{p_b}\right)$, 
with $W$ the Weyl tensor, vanishes if $\mg$ is conformally flat.
\end{itemize}
\end{prop}

\begin{proof}
Let $S\in\sS^0_{k,\k}$ be a conformal supercharge.
There exists $K\in\bigoplus_{\ell}\sT^0_{k-\ell,\k+2\ell}$ such that $\fS^0_\nabla(K)=S$. 
Since $\sT^0=(\bDel,R)\oplus\ker T\cap\ker\bDel^*$, $K$ splits accordingly as $K=K_0+K_1$. 
In view of \eqref{cSCovariant}, the relation $K_0\in(\bDel,R)$ implies that 
$\fS^0_\nabla(K_0)\in(\bDel,R)$ is a trivial supercharge. Hence, $\fS^0_\nabla(K_1)$
is a conformal supercharge.
As $K_1\in\ker T\cap\ker\bDel^*$, we can use the Equivalence \eqref{SuperConf}.
By collecting the terms of $p$-degree $k+1$ in the right hand side,
we deduce that the component of $K_1$ with $p$-degree equal to $k$
is a conformal Killing hook tensor. As a result, we obtain that $\ve_{k,\k}(S)\in(\bDel,R)\oplus\cKh$.

Let $K\in\cKh$. Then, $\bdel^*K=0$ and $\Pi_{01}G^\nabla K=0$.
By Equivalence \eqref{SuperConf}, $\fS^0_\nabla(K)$ is a conformal supercharge if and only if
$\Pi_{01}(\bd^\nabla)^2K=0$. Moreover, we have 
$\Pi_{01}\circ(\bd^\nabla)^2=
\Pi_{01}\circ\left(\xi^k\xi^lR^a_{~bkl}p_a\partial_{p_b}\right)=
\Pi_{01}\circ\left(\xi^k\xi^lW^a_{~bkl} p_a\partial_{p_b}\right)$ 
since the traces of the Riemann tensor are killed by the projection $\Pi_{01}$.
\end{proof}

\begin{thm}
Let $(M,\mg)$ be a conformally flat manifold. The conformally equivariant superization 
descends as a linear isomorphism $\fS^0:\sT^0/(\bDel,R)\longrightarrow
\sS^0/(\bDel,R)$ and provides an isomorphism of $\fkg$-modules 
\begin{equation}\label{Iso:K-SC}
\fS^0:\cKh\longrightarrow
\cS\cC,
\end{equation}
between conformal Killing hook tensors and (equivalence classes of) conformal supercharges.
\end{thm}

\begin{proof}
By Eq.\ \eqref{cSCovariant}, $K\in(\bDel,R)$ implies $\fS^0(K)\in(\bDel,R)$.
Hence $\fS^0$ descends as  a linear isomorphism $\fS^0:\sT^0/(\bDel,R)\longrightarrow
\sS^0/(\bDel,R)$.
Propositions \ref{Curved-Flat} and \ref{Prop:KH-SC curved} allow to conclude.
\end{proof}

\begin{rmk}
The condition $\Pi_{01}G^\nabla K=\frac{\hbar}{\bi}\Pi_{01}(\bd^\nabla)^2 K$
reduces in dimension $4$, for a symmetric conformal Killing $2$-tensor $K$, to the condition (A1) in \cite{ABB14}.
The authors of \cite{ABB14} prove the following.
A symmetric conformal Killing $2$-tensor $K$ is the principal symbol of a higher symmetry of the Dirac operator (see below)
if and only if  $K$ satisfies condition (A1) and $K$ is the principal symbol 
of a higher symmetry for the conformal Laplacian (condition (A0) in \cite{ABB14}).
\end{rmk}

%%%%%%%%%%%%%%%%%%%%%%
\subsection{Higher symmetries of the Dirac operator}

The Dirac operator $\Dirac\in\Dslm$ is a $G$-invariant operator 
for $\l=\frac{n-1}{2n}$ and $\m=\frac{n+1}{2n}$, where $G$ is the (local)
conformal Lie group of~$(M,\mg)$. 
We introduce the higher symmetries of $\Dirac$ following the Laplacian case. 

\begin{defi}
Let $\l=\frac{n-1}{2n}$, $\m=\frac{n+1}{2n}$. A higher symmetry of $\Dirac$ is a differential operator $D_1\in\sD^{\l,\l}$, such that $\Dirac D_1=D_2\Dirac$, for some $D_2\in\sD^{\m,\m}$.
\end{defi} 

The space of higher symmetries of $\Dirac$ is a subalgebra of $\sD^{\l,\l}$ 
containing the ideal of trivial symmetries $(\Dirac)=\{A\Dirac \vert \, A\in\sD^{\m,\l}\}$.
We denote by $\cA^\Dirac$ the  quotient algebra of the higher symmetries of $\Dirac$ by the trivial ones. 
This is exactly the kernel of the following conformally invariant operator
\begin{eqnarray}\nonumber
\HSQ:\sD^{\l,\l}/(\Dirac) & \rightarrow & \sD^{\l,\m}/(\Dirac)\\ \label{HSQ}
 \left[D\right] &\mapsto & [\Dirac D]
\end{eqnarray}
where $\HSQ$ stands for {\it Quantum Higher Symmetries}. 
Since $[\Dirac D]$ is equal to the class of the commutator $[\Dirac,D]$, the Property \eqref{PcpalSymbol} 
implies that
\begin{equation}\label{HSQtoQ}
\xymatrix{
\sD^{\l,\l}_{[k]}/(\Dirac)\ar[d]^{\sigma_k}\ar[rr]^{\HSQ} && \sD^{\l,\m}_{[k+1]}/(\Dirac)\ar[d]^{\sigma_{k+1}}\\
\sS^0_{[k]}/\sigma_k\big((\Dirac)\big) \ar[rr]_{\{\bDel,\cdot\}}&& \sS^{\frac{1}{n}}_{[k+1]}/\sigma_{k+1}\big((\Dirac)\big)
}
\end{equation}
is a commutative diagram of $G$-modules, with $\sigma_k, \sigma_{k+1}$ principal Hamiltonian symbol maps.
In the conformally flat case, the conformally equivariant quantization allows to 
invert the vertical maps in Diagram \eqref{HSQtoQ} and, as proved below, 
provides a correspondence
between conformal supercharges and higher symmetries of $\slashed{D}$.
In particular, it allows to recover the classification of first order higher symmetries 
of $\slashed{D}$ \cite{BKr04}.
Nevertheless the correspondence is not bijective: $\g(\vol_\mg)$ is a higher symmetry
but $\vol_\mg$ is not a conformal supercharge and,
conversely, all the elements in $\sT^0_{1,n-1,0;01}$ 
are conformal Killing hook tensors and non-trivial conformal supercharges, 
but they are quantized as trivial symmetries.

\begin{thm}
Let $\l=\frac{n-1}{2n}$, $\m=\frac{n+1}{2n}$ and $(M,\mg)$ be a conformally flat spin manifold of even dimension. We have the following isomorphism of $G$-modules
$$
\cQ^{\l,\l}:\cS\cC/\sT^0_{1,n-1,0;01}\longrightarrow\cA^\Dirac/\bbC\cdot\g(\vol_\mg)
%\qquad\text{and}\qquad \cQ^{\l,\l}\circ\fS^0:\cKh/\sT^0_{1,n-1,0;01}\rightarrow\cA^\Dirac/\bbR\cdot\g(\vol_\mg)
$$
which essentially establishes a correspondence between %conformal Killing hook tensors, 
conformal supercharges and higher symmetries of $\Dirac$. Those of first order are given by
\begin{eqnarray*}
\cQ^{\l,\l}\circ\fS^0(\bdel^* K) &=& \frac{\hbar}{\bi(\sqrt{2})^{\k}} \left(\mg^{ij}\g(K_i)\nabla^\l_{j}-\frac{\k+1}{\k+2}\g(\bd^\nabla K) +\frac{n-\k-1}{2(n-\k)}\g\left(\bd^{*\nabla}  K\right) \right),
\end{eqnarray*}
where $\k$ runs over $0,\ldots,n-1$ and $K$ runs over the space of conformal Killing $\k+1$-forms.
\end{thm}   

\begin{proof}
We introduce the submodule 
$(\bDel,R)_*:=(\bDel,R)\oplus\sT^\d_{1,n-1,0;01}$
of the $G$-module $\sT^\d$.
We will need the three following lemmas.

\begin{lem}\label{lem:ideals}
The $G$-module $(\bDel,R)_*$ is preserved by the conformally equivariant superizations 
$\fS^{\d}$. If $\l=\frac{n-1}{2n}$ and $\m=\frac{n+1}{2n}$, 
we have $\cQ^{\l,\l}\left((\bDel,R)_*\right)=(\Dirac)$ and there exists 
a conformally equivariant quantization $\cQlm$ such that $\cQ^{\l,\m}\left((\bDel,R)_*\right)=(\Dirac)$.
\end{lem}

\begin{proof}
The Formula \eqref{EqrS} leads to the result concerning 
the conformally equivariant superizations. Let $\l'=\l,\m$. According to 
Corollary \ref{cor:existQ}, the conformally equivariant quantizations $\cQ^{\l,\l}$  
and $\cQ^{\m,\l'}$ exist and are unique. The conformal invariance of $\Dirac\in\Dslm$ 
ensures that $P\bDel\mapsto\cQ^{\m,\l'}(P)\Dirac\in\sD^{\l,\l'}$ defines 
a conformally equivariant map on $(\bDel)$. This conformally extends to $(\bDel,R)$ 
via $\fS^0(P R)\mapsto\cQ^{\m,\l'}(P\bDel)\Dirac$ for $P\notin(\bDel)$. 
If $\l'=\l$ this map coincides with $\cQ^{\l,\l}$ by uniqueness. If $\l'=\m$, 
this defines a conformally equivariant quantization on $(\bDel,R)$, 
which can be extended to $\cS^{\frac{1}{n}}=(\bDel,R)\oplus\fS^{\frac{1}{n}}\left(\ker\bDel^*\cap\ker T\right)$. 
Indeed, combining Theorems \ref{ExistenceSuper} and \ref{thmConfInv}, 
the only obstructions to existence of $\cQlm$ are given by the conformally invariant operators
$\fS^{\frac{1}{n}}\circ GT\circ (\fS^{\frac{1}{n}})^{-1}$ and 
$\fS^{\frac{1}{n}}\circ G\bdel\bDel^*\circ (\fS^{\frac{1}{n}})^{-1}$, 
which vanish on $\fS^{\frac{1}{n}}\left(\ker\bDel^*\cap\ker T\right)$. 
Using Theorem \ref{thm_equi_adaptees} and the obtained factorization formul{\ae} 
for $\cQ^{\l,\l'}$, we deduce that $\cQ^{\l,\l'}\left((\bDel,R)_*\right)\subset(\Dirac)$.
To prove the converse inclusion, it suffices to show that the principal symbol
of $D\Dirac$ pertains to $(\bDel,R)_*$, for all $D\in\sD^{\l,\l'}$.
Suppose $D$ is exactly of order $k$.
There are three cases.
If $\sigma_k(D)\bDel\neq 0$, we get $0\neq\sigma_{k+3}(D\Dirac)\in(\bDel)$ by Eq.\ \eqref{PcpalSymbol}.
If $\sigma_k(D)\bDel= 0$ and $\sigma_k(D)\in\sS^{\l'-\l}_{0,k}$,
then we have $k=n$, $D\Dirac\in\sD^{\l,\l'}_{[n+1]}$ and 
$0\neq\sigma_{n+1}(D\Dirac)\in\sT^{\l'-\l}_{1,n-1,0;01}$.
In the remaining case, Proposition \ref{Prop:proj} applies and we obtain $\sigma_k(D)\in(\bDel)$.
Hence, $0\neq\sigma_{k+1}(D\Dirac)\in(R)$ and the result is proved.
\end{proof}

\begin{lem}\label{ConfInvT}
Let $A:\sT^0/(\bDel,R)_*\rightarrow\sT^{\frac{1}{n}}/(\bDel,R)_*$ be a $G$-invariant operator. 
Then, on $\sT^0_{k,\k}/(\bDel,R)_*$ (with $k\neq 0$), the operator $A$ is proportional to 
a linear combination of $\bdel^*$ and $\bdel\Pi_{01}G$. 
\end{lem}

\begin{proof}
By Theorem \ref{DiagCt}, the space $\sT^0/(\bDel,R)_*$ is a submodule of
$\bigoplus_{k,\k,\b}\sT^0_{k,\k,0;0\b}$.
In view of Theorem \ref{thmConfInv}, if $k$ is non-vanishing, 
the spaces of $G$-invariant linear operators with source spaces 
$\sT^0_{k,\k,0;00}$ and $\sT^0_{k,\k,0;01}$ are generated by $\bdel^*$ and $G_{\bf 0}\bdel$. 
By Proposition \ref{Prop:proj}, we have $\Pi_{\bf 0}=\frac{1}{\cE+\Sigma}\bdel\bdel^*$
and $\Pi_{01}=\frac{1}{\cE+\Sigma}\bdel^*\bdel$ on $\sT^0_{>0}/(\bDel,R)_*$.
Hence, on the space $\sT^0_{k,\k,0;01}$, we have 
$G_{\bf 0}\bdel=\frac{1}{\cE+\Sigma}\bdel G\bdel^*\bdel=\bdel G \frac{\cE+\Sigma}{\cE+\Sigma+1}$
and $\bdel\Pi_{01}G=\frac{1}{\cE+\Sigma}\bdel\bdel^*\bdel G=\bdel G$.
The operators $\bdel\Pi_{01}G$ and $G_{\bf 0}\bdel$ are proportional,
this concludes the proof.
\end{proof}

By Lemma \ref{lem:ideals}, the maps $\fS^\d:\sT^\d/(\bDel,R)_*\rightarrow\sS^\d/(\bDel,R)_*$ 
are well-defined. Since $\{\bDel,(\bDel,R)_*\}\subset (\bDel,R)_*$,
the map $\{\bDel,\cdot\}:\sS^0/(\bDel,R)_*\rightarrow\sS^{\frac{1}{n}}/(\bDel,R)_*$
is also well-defined.

\begin{lem}\label{QtoA}
The operator $A$ defined by the following commutative diagram
\begin{equation}\label{HSC}
\xymatrix{
\sS^0/(\bDel,R)_*\ar[rr]^{\{\bDel,\cdot\}} && \sS^{\frac{1}{n}}/(\bDel,R)_* \\
\sT^0/(\bDel,R)_* \ar[u]^{\fS^0}\ar[rr]_{A}&& \sT^{\frac{1}{n}}/(\bDel,R)_*\ar[u]_{\fS^{\frac{1}{n}}}
}
\end{equation}
satisfies $A=-\frac{\bi}{\hbar}\bdel^*+\frac{1}{\cE+\Sigma}\bdel\Pi_{01}G$.
\end{lem}

\begin{proof}
We have to prove that $\{\bDel,\fS^0(K)\}-\fS^{\frac{1}{n}}(A K)\in(\bDel,R)_*$ 
for $K$ in the $G$-module $\sT^0/(\bDel,R)_*\leq\ker\bDel^*\cap\ker T$. 
By Theorem \ref{DiagCt}, the latter space splits into $\ker\bdel + \ker\bdel^*$ 
and leads to two cases. As $(M,\mg)$ is conformally flat, we can deduce 
from Computation \eqref{Compute:HSC} that 
\begin{eqnarray*}
\{\bDel,\fS^0(K)\}-(\frac{\bi}{\hbar}\bdel^*+d) (K)\in(\bDel,R)_* &\text{if}& K\in\ker\bdel,\\ 
\{\bDel,\fS^0_\nabla(K)\}-\frac{1}{\cE+\Sigma}\bdel G K\in(\bDel,R)_* & \text{if} &  K\in\ker\bdel^*. 
\end{eqnarray*}
By the proof of Lemma \ref{ConfInvT}, we have 
$\bdel G K-\bdel\Pi_{01}G K\in(\bDel,R)_*$ if $K\in\ker\bdel^*$. 
The formula \eqref{EqrS} giving the superization allows then to conclude in both cases. 
\end{proof}

We are now ready to prove the theorem. According to Lemma \ref{lem:ideals}, 
the quantizations $\cQ^{\l,\l'}$ descend to the quotient spaces as follows,
$\cQ^{\l,\l'}:\sS^\d/(\bDel,R)_*\rightarrow\sD^{\l,\l'}/(\Dirac)$, for $\l'=\l+\d$ and $\d=0,1/n$. 
Hence, we get the following commutative diagram of $G$-modules
\begin{equation}\label{HSQ-HSC}
\xymatrix{
\sD^{\l,\l}/(\Dirac)\ar[rr]^{\HSQ} && \sD^{\l,\m}/(\Dirac)\\
\sS^0/(\bDel,R)_*\ar[u]^{\cQ^{\l,\l}}\ar[rr]^{\HSC} && \sS^{\frac{1}{n}}/(\bDel,R)_* \ar[u]_{\cQlm}\\
\sT^0/(\bDel,R)_* \ar[u]^{\fS^0}\ar[rr]_{A}&& \sT^{\frac{1}{n}}/(\bDel,R)_*\ar[u]_{\fS^{\frac{1}{n}}}
}
\end{equation}
where $\HSC$ and $A$ are conformally invariant operators. The Diagram \eqref{HSQtoQ} 
leads to $\HSC=\{\bDel,\cdot\}+B$, where $B$ does not rise the Hamiltonian degree, 
contrary to $\{\bDel,\cdot\}$ which rises it by one. On $\sT^0_{>0}$, by Lemma \ref{ConfInvT}, 
we know the form of $A$ and together with Lemma \ref{QtoA} we deduce that 
the only possibility is $\HSC=\{\bDel,\cdot\}$ and $A=-\frac{\bi}{\hbar}\bdel^*+\frac{1}{\cE+\Sigma}\bdel\Pi_{01}G$.
On $\sT^0_{0,\k}$, with $\k<n$, the same holds. On $\sT^0_{0,n}$, we have $\HSC=A=\bd^*$.

In view of the latter diagram, the kernels of $\HSC$ and $\HSQ$
 are isomorphic via $\cQ^{\l,\l}$. 
By definition, $\ker\HSQ=\cA^\Dirac$ and 
the kernel of $\{\bDel,\cdot\}$ on $\sT^0/(\bDel,R)$ is $\cS\cC$.
Since the kernel of $\bd^*$ on $\sT^0_{0,n}$ is $\bbC\cdot\vol_\mg$,
we obtain that $\ker\HSC=\bbC\cdot\vol_\mg\oplus\cS\cC/\sT^0_{1,n-1,0;01}$.  

It remains to compute $\cQ^{\l,\l}\circ\fS^0(\bdel^* K)$ for $K$ a conformal Killing $\k+1$-form. 
The result follows from Eqs.\ \eqref{cSCovariant}-\eqref{cQCovariant}
and from the relations $\bdel\bdel^* K= (\k+1)K$, $D^\nabla\bdel^* K=\bd^* K$. 
%The former relation is a consequence of Proposition \ref{Prop:proj} 
%and the latter follows from the coordinate expression of the involved operators.
\end{proof}

%%%%%%%%%%%%%%%%%%%%%%%%%%%%%%%%%%%%%%%%%%%%%%%%%
 \subsection*{Acknowledgements}
%%%%%%%%%%%%%%%%%%%%%%%%%%%%%%%%%%%%%%%%%%%%%%%%%
 It is a pleasure to acknowledge Christian Duval for his essential guidance in our investigation of geometric and conformally equivariant quantizations of $(\cM,d\a)$. Special thanks are due to Valentin Ovsienko for his constant interest in this work and to Josef \v{S}ilhan for discussions on higher symmetries of the Dirac operator.

\appendix
%%%%%%%%%%%%%%%%%%%%%%%%%%%%%%%%%%%%%%%%%%%%%%%%%
%%%%%%%%%%%%%%%%%%%%%%%%%%%%%%%%%%%%%%%%%%%%%%%%%
 \section{}
%%%%%%%%%%%%%%%%%%%%%%%%%%%%%%%%%%%%%%%%%%%%%%%%%
%%%%%%%%%%%%%%%%%%%%%%%%%%%%%%%%%%%%%%%%%%%%%%%%%
We collect here informations on the $13$ generators of the $\rE(p,q)$-invariant operators on $\D(\cM)$ introduced in Proposition \ref{InvIso}.
  
\subsection{}In the following table we recall their definitions together with their interpretation
as operators on $\Ga(\cS TM\otimes\Lambda T^*M)$, in the flat case $(M,\mg)=(\bbR^{p,q},\eta)$. 

\begin{equation}\label{InterpretationInv}
\begin{array}{|c|c||c|c|}
				\hline	
R=\eta^{ij}p_ip_j	& \text{metric} & D=\partial_{p_i}\partial_i &\text{divergence}\\[3pt]\hline
\cE=p_i\partial_{p_i} & \text{ Even Euler operator} & G=\eta^{ij}p_i\partial_j & \text{gradient}\\[3pt]\hline
T=\eta_{ij}\partial_{p_i}\partial_{p_j} & \text{trace} & L=\eta^{ij}\partial_i\partial_j & \text{Laplacian}\\[3pt]\hline
\hline
\Sigma=\xi^i\partial_{\xi^i}& \text{ Odd Euler operator} & & \\[3pt]\hline
\bdel=\eta_{ij}\xi^i\partial_{p_j} & \text{Koszul differential} & \bd=\xi^i\partial_i & \text{de Rham differential} \\[3pt]\hline
\bdel^*=\eta^{ij}p_i\partial_{\xi^j} & \text{Koszul codifferential}& \bd^*=\partial_{\xi^i}\partial_i & \text{de Rham codifferential} \\[3pt]\hline
\bDel=p_i\xi^i &  \text{Berezin differential} & &\\
&\text{or symbol of } \Dirac & &\\[3pt]\hline
\bDel^*=\partial_{\xi^i}\partial_{p_i} & \text{Berezin codifferential}& & \\[3pt]\hline
\end{array}
\end{equation}  

\subsection{}We compute the action of the inversion $\bar{X}_i$, see \eqref{AlgLieConf}, on the five generators of $\fh(2|1,1)$ viewed as operators $A:\sT^\d\rightarrow\sT^{\d'}$ with $\d'$ chosen according to Table \eqref{tableInv}. Explicitly, this action reads as $[A,\bbL_{\bar{X}_i}^*]:=A\bbL_{\bar{X}_i}^\d-\bbL_{\bar{X}_i}^{\d'}A$ and we get
\begin{align}\nonumber%\label{DXi}
\left[D,\bbL^\d_{\bar{X}_i}\right] &= 2\big(2\cE+n(1-\d)\big)\partial_{p_i}-2p_iT+2\bdel\partial_{\xi^i}-2\xi_i\bDel^*,\\[3pt] \nonumber%\label{G0Xi} 
\left[G,\bbL_{\bar{X}_i}^*\right] &= -2n\d p_i+2R\partial_{p_i}+2\bDel\partial_{\xi^i}-2\xi_i\bdel^*,\\[3pt] \label{Commutateur}%\label{L0Xi} 
\left[L,\bbL_{\bar{X}_i}^*\right] &= 2 \big(2\cE+n(1-2\d)\big)\partial_i+4G\partial_{p_i}-4p_iD+4\bd\partial_{\xi^i}-4\xi_i\bd^*,\\[3pt] \nonumber%\label{Ga0Xi}
\left[\bd,\bbL_{\bar{X}_i}^*\right] &= 2\big(\cE+\Sigma-1-n\d\big)\xi_i+2\bDel\partial_{p_i}-2p_i\bdel,\\[3pt]
\nonumber%\label{La0Xi}
\left[\bd^*,\bbL_{\bar{X}_i}^*\right] &= 2\big(\cE-\Sigma-1+n(1-\d)\big)\partial_{\xi^i}-2p_i\bDel^*+2\bdel^*\partial_{p_i}.
\end{align}
We introduce $\Pi_{\bf 0}$, the conformally invariant projection on $\ker T\cap\ker\bDel^*\cap\ker\bdel$, and denote by an index ${\bf 0}$ the five generators of $\fh(2|1,1)$ restricted and corestricted to that space.  Then, the action of the inversion $\bar{X}_i$ on their powers, acting on $\sT^\d_{k,\k,0;00}$, reads as
\begin{align}\nonumber%\label{D0Xi}
\left[D_{\bf 0}^d,\bbL^\d_{\bar{X}_i}\right] &= 2d\big(2k-d+n(1-\d)\big)\partial_{p_i}D_{\bf 0}^{d-1},\\[3pt]\nonumber% \label{G0Xi} 
\left[G_{\bf 0}^g,\bbL_{\bar{X}_i}^*\right] &= -2g\big(g+n\d\big)\Pi_{\bf 0}p_iG_{\bf 0}^{g-1},\\[3pt] \label{L0Xi} 
\left[L_{\bf 0}^\ell,\bbL_{\bar{X}_i}^*\right] &= 2\ell \Big(\big(2(k-\ell)+n(1-2\d)\big)\partial_i+4\big(G\partial_{p_i}+\bd\partial_{\xi^i}-p_iD-\xi_i\bd^*\big)\Big)L_{\bf 0}^{\ell-1},\\[3pt] \nonumber% \label{Ga0Xi}
\left[\bd_{\bf 0},\bbL_{\bar{X}_i}^*\right] &= 2\big(k+\k-n\d\big)\Pi_{\bf 0}\xi_i,\\[3pt]
\nonumber%\label{La0Xi}
\left[\bd^*_{\bf 0},\bbL_{\bar{X}_i}^*\right] &= 2\big(k-\k+n(1-\d)\big)\Pi_{\bf 0}\partial_{\xi^i}.
\end{align}

\subsection{}Recall that $E=\cE+\frac{n}{2}$, and $\mathsf{\Sigma}=\Sigma-\frac{n}{2}$.
We sum up all the commutation relations between the previous $13$ operators, they generate the super Lie algebra $\spo(2|1,1)\ltimes\fh(2|1,1)$. 
\begin{equation}\label{CommRel}
\begin{array}{|c||c|c|c|c|c|c|c|c||c|c|c|c|c|}
				
				\hline
 & R & E & T    & \mathsf{\Sigma} & \bdel & \bdel^*  & \bDel & \bDel^* & D & G & L & \bd & \bd^{*} \\[3pt]\hline
	\hline
R& 0 & -2R & -4E& 0  & -2\bDel & 0 &  0     &  -2\bdel^* & -2G& 0 & 0& 0   &  0  \\[3pt]\hline
E& 2R& 0 & 2T  &  0     & -\bdel & \bdel^*  & \bDel & -\bDel^* & -D& G & 0 & 0   &  0  \\[3pt]\hline
T&4E & -2T& 0& 0   & 0   &  2\bDel^* & 2\bdel  &  0      & 0 & 2D& 0 & 0   &  0    \\[3pt]\hline
\mathsf{\Sigma}&	0& 0  &  0  & 0 & \bdel &  -\bdel^*   & \bDel & -\bDel^* & 0 & 0 & 0 & \bd & -\bd^{*}  	 \\[3pt]\hline
\bdel& 2\bDel&\bdel&0& -\bdel & 0&E+\mathsf{\Sigma}&  0     & T       & 0 &\bd& 0 & 0   & D	  \\[3pt]\hline
\bdel^*& 0& -\bdel^*&-2\bDel^*&\bdel^*&E+\mathsf{\Sigma}& 0 &  R     &  0      &\bd^*&0&0& G& 0 \\[3pt]\hline
\bDel&0&-\bDel &-2\bdel& -\bDel&0& R &0& E-\mathsf{\Sigma}& -\bd& 0& 0	&0& G		\\[3pt]\hline
\bDel^*& 2\bdel^* &\bDel^* & 0& \bDel^* & T & 0 & E-\mathsf{\Sigma}& 0 & 0 & \bd^{*} & 0 & D & 0\\[3pt]\hline
\hline
D& 2G & D& 0 & 0 & 0 & \bd^{*} & \bd & 0 & 0 & L & 0 & 0& 0 				\\[3pt]\hline
G& 0 & -G& -2 D	& 0 & -\bd & 0 & 0& -\bd^*& -L & 0& 0 & 0& 0			\\[3pt]\hline
L&0 &0&0 &0&0 &0&0 &0&0 &0&0 &0&0  				\\[3pt]\hline
\bd&0 &0&0 &-\bd  & 0& G & 0& D& 0 &0 &0 &0 &L				\\[3pt]\hline
\bd^*&0 &0&0 &\bd^*& D &0& G& 0 &0&0&0 &L&0		\\[3pt]\hline
\end{array}
\end{equation}  \\
We denote by a zero index the operators $\bdel^*,\bdel,\bDel,\bDel^*$ restricted and corestricted to the kernel of the operator $T$. They satisfy the following commutation relations:
\begin{equation}\label{CommutationkerT}
\begin{array}{|c|c|c|c|c|}
				\hline	
	[\cdot,\cdot]				& \bDel_0 & \bdel^*_0 & \bdel_0 & \bDel^*_0\\[3pt]\hline
\bDel_0	& 0				& 0		& -4\sfc\bDel_0\bdel_0		& (n+\cE-\Sigma)-4\sfc\bdel^*_0\bdel_0\\[3pt]\hline 
\bdel^*_0		&	0				&	0		&	\Sigma+\cE-4\sfc\bDel_0\bDel^*_0 & -4\sfc\bdel^*_0\bDel^*_0\\[3pt]\hline 
\bdel_0		& -4\sfc\bDel_0\bdel_0	&	\Sigma+\cE-4\sfc\bDel_0\bDel^*_0 & 0 & 0\\[3pt]\hline 
\bDel^*_0 &(n+\cE-\Sigma)-4\sfc\bdel^*_0\bdel_0& -4\sfc\bdel^*_0\bDel^*_0&0&0\\ 
\hline 
\end{array}
\end{equation}
where $\sfc=\frac{1}{2\left(n+2(\cE-1) \right)}$ comes from the coefficient of $RT$ in $\Pi_0$, but with $\cE-1$ instead of $\cE$ as the commutation with $\bdel$ or $\bDel^*$ lowers by $1$ the $p$-degree.

% ==================================================================
% BIBLIOGRAPHIE
% Gestion am�lior�e de la bibliographie
\bibliographystyle{plain}
\bibliography{Biblio}
 
\end{document}